\documentclass[11pt]{article}
\usepackage[a4paper, top=2.5cm, bottom=2.5cm, left=2.5cm, right=2.5cm]
{geometry}

\usepackage{graphicx}
\usepackage{amsmath}
\usepackage{amsthm}
\usepackage{amsfonts}
\usepackage{amssymb}
\usepackage{natbib}
\usepackage{makecell}
\usepackage{algorithmic}
\usepackage[linesnumbered,algoruled,boxed,lined]{algorithm2e}
\usepackage{xcolor}
\usepackage{subcaption}
\usepackage{soul}
\usepackage{multirow}
\usepackage{hyperref}

\newcommand{\EFX}{\text{EFX}}

\newtheorem{definition}{Definition}
\newtheorem{proposition}{Proposition}
\newtheorem{corollary}{Corollary}
\newtheorem{theorem}{Theorem}
\newtheorem{claim}{Claim}
\newtheorem{remark}{Remark}
\newtheorem{lemma}{Lemma}
\newtheorem*{example}{Example}

\title{A Complete Landscape of EFX Allocations on Graphs: \\ Goods, Chores and Mixed Manna\thanks{A preliminary version of this paper appeared in the proceedings of IJCAI 2024 \citep{DBLP:conf/ijcai/0027WL024}. The authors are listed alphabetically and Yu Zhou is the corresponding author.}}

\author{
Bo Li$^1$ \hspace{15pt} 
Minming Li$^2$ \hspace{15pt} 
Tianze Wei$^2$ \hspace{15pt} 
Zekai Wu$^1$ \hspace{15pt} 
Yu Zhou$^3$\\ \small
$^1$Department of Computing, The Hong Kong Polytechnic University \\\small
\texttt{comp-bo.li@polyu.edu.hk, zekai-candyore.wu@connect.polyu.hk} \\ \small
$^2$Department of Computer Science, City University of Hong Kong \\ \small
\texttt{minming.li@cityu.edu.hk, t.z.wei-8@my.cityu.edu.hk} \\ \small
$^3$College of Information Sciences and Technology, The Pennsylvania State University \\ \small
\texttt{yqz6094@psu.edu}
}

\date{}

\begin{document}
\maketitle

\begin{abstract}
We study envy-free up to any item (EFX) allocations on simple graphs where vertices and edges represent agents and items respectively. 
An agent (vertex) is only interested in items (edges) that are incident to her and all other items always have zero marginal value to her. 
Christodoulou et al. [EC, 2023] first proposed this setting and studied the case of goods where every item has non-negative marginal values to every agent.
In this work, we significantly generalize this setting and provide a complete set of results by considering the allocation of arbitrary items that can be goods, chores, or mixed manna under doubly monotone valuations with a mild assumption. 
For goods, we complement the results by Christodoulou et al. [EC, 2023] by considering another weaker notion of EFX in the literature and showing that an orientation  -- a special allocation where each edge must be allocated to one of its endpoint agents -- that satisfies the weaker notion always exists and can be computed in polynomial time, contrary to the stronger notion for which an orientation may not exist and determining its existence is NP-complete. 
For chores, we show that an envy-free allocation always exists, and an EFX orientation may not exist but its existence can be determined in polynomial time. 
For mixed manna, we consider the four notions of EFX in the literature. 
We prove that an allocation that satisfies the strongest notion of EFX may not exist and determining its existence is NP-complete, while one that satisfies any of the other three notions always exists and can be computed in polynomial time.     
We also prove that an orientation that satisfies any of the four notions may not exist and determining its existence is NP-complete. 
\end{abstract}

\section{Introduction}
Fair allocation of indivisible items has been broadly studied in the research fields of computer science, economics, and mathematics in the past few decades. 
One of the most compelling and natural fairness notions is \textit{envy-freeness} (EF), which requires that every agent prefers her own bundle to any other agent's bundle. 
Though envy-freeness can always be satisfied for divisible items \citep{DBLP:conf/focs/AzizM16, DBLP:conf/soda/DehghaniFHY18}, it is too demanding for indivisible items. 
An EF allocation does not exist even for the simple instance where there are two agents and one indivisible item with non-zero marginal values to both agents. 

The fact that envy-freeness is hard to satisfy necessitates the study of its relaxations, the most popular one among which is \textit{envy-freeness up to any item} (EFX). 
EFX requires that any envy could be eliminated by virtually removing any item that the envious agent likes from the envied agent's bundle or any item that the envious agent dislikes from her own bundle.  
As remarked by \cite{DBLP:conf/ec/CaragiannisGH19}: \textit{``Arguably, EFX is the best fairness analog of envy-freeness for indivisible items.''}
Despite significant effort in the literature, the existence of EFX allocations still remains an open problem for indivisible items.
Only a few special cases are known to admit EFX allocations. 
For the case of goods where every item is liked by every agent, \cite{plaut2020almost} showed that EFX allocations exist when there are only two agents. 
\cite{chaudhury2020efx} and \cite{DBLP:conf/esa/Mahara21} complemented this result by showing the existence of EFX allocations when there are only three agents, and when the number of goods is at most three greater than the number of agents, respectively. 
\cite{plaut2020almost} also showed that EFX allocations exist when the valuations are identical or additive with identical ordering. 
\cite{DBLP:journals/tcs/AmanatidisBFHV21} extended this result by showing the existence of EFX allocations when the valuations are bi-valued. 
Other restricted valuations for which EFX allocations always exist include submodular dichotomous valuations where the marginal value of any item is either 0 or 1 \citep{DBLP:conf/aaai/BabaioffEF21}, and lexicographic valuations where the value of a single item is greater than the total value of all items that are less preferred \citep{DBLP:conf/aaai/HosseiniSVX21}. 
For the case of chores where every item is disliked by every agent, even less is known. For example, \cite{DBLP:conf/www/0037L022} and \cite{gafni2023unified} showed that EFX allocations exist for ordered instances and leveled preference instances, respectively. 

Recently, \cite{christodoulou2023fair} studied EFX allocations on graphs where vertices correspond to agents and edges correspond to indivisible goods. 
An agent (vertex) is only interested in the goods (edges) that are incident to her and all other edges always have zero marginal values to her. 
Thus, each good is liked by exactly two agents in their setting.
As motivated in \cite{christodoulou2023fair}, a direct application of this setting is the allocation of geographical resources, for instance, natural resources among countries on the boundaries, working offices among research groups, and public areas among communities in a region, etc. 
\cite{christodoulou2023fair} proved that EFX allocations always exist and can be computed in polynomial time for arbitrary graphs. 
Remarkably, this is one more rare case with more than three agents for which an EFX allocation is guaranteed to exist. 
They also considered a more restricted scenario where each edge must be allocated to one of its endpoint agents. 
In this scenario, an allocation is also called an {\em orientation}. 
Unfortunately, \cite{christodoulou2023fair} proved that an EFX orientation may not exist, and determining whether it exists or not is NP-complete. 

Besides goods and chores, recent years have also seen a rapidly growing interest in the case of mixed manna in the literature of fair division. 
A mixed manna contains items that are goods for some agents but chores for others. 
Practically, the setting of mixed manna can model the scenarios
where agents have different opinions on items.
Many real-world scenarios involve the allocation of mixed manna.
For example, when the items are paid jobs, they are goods for some people because completing them can
bring extra revenue; 
however, they can be chores for some people who do not care much about this amount of money and would like to save time for other matters.
The case of mixed manna is also a typical setting where the valuations are not monotone.  
\cite{christodoulou2023fair}'s graphic nature also appears in the setting of mixed manna.
For example, in sports games, each match (that can be viewed as an item) involves two teams (that can be viewed as the agents) and has to be hosted by one of them (i.e., home or away). 
Hosting a match might be a good for some teams as they can make profit and might be a chore as they cannot cover the expenses.
Furthermore, the graph orientation setting can use the topology to indicate
who are capable of completing what jobs (edges), so that the jobs can only be allocated to people (incident
vertices) who are able to do them. The allocation setting
can model the case when people really do not have any
cost or benefit on the items they are not incident to.

\subsection{Our Problem and Results}

\begin{table}[!tb]
\centering
\begin{tabular}{|c|c|c|c|}
    \hline
     Item Types & Notions & Orientation & Allocation \\
    \hline
    \multirow{3}*{Goods}& $\EFX^0$ & \makecell{may not exist, NP-c\\ (\citep{christodoulou2023fair})} & \makecell{always exist, P\\ (\citep{christodoulou2023fair})}\\
    \cline{2-4}
    & $\EFX^+$ & \makecell{always exist, P \\ (Proposition \ref{prop:goods_orien:EFX+})} & \makecell{always exist, P \\ (Corollary \ref{cor:goods_orien:EFX+})} \\
    \hline
    \multirow{3}*{Chores}& $\EFX_0$ & \makecell{may not exist, P\\ (Theoreom \ref{thm:chores_orien:EFX0} and \\ \citep{hsu2025polynomial})} & \makecell{always exist, P\\ (Proposition \ref{prop:chores_allo})}\\
    \cline{2-4}
    & $\EFX_-$ & \makecell{may not exist, P \\ (Proposition \ref{prop:chores_orien:EFX-})} & \makecell{always exist, P\\ (Proposition \ref{prop:chores_allo})} \\
    \hline
    \multirow{7}*{\makecell{Mixed\\Manna}}& $\EFX^0_0$ & \makecell{may not exist, NP-c\\ (Corollary \ref{coro:orientation:EFX00_EFX0-})} & \makecell{may not exist, NP-c\\ (Theorem \ref{thm:allocation:EFX00})}\\
    \cline{2-4}
    &$\EFX^0_-$ & \makecell{may not exist, NP-c\\ (Corollary \ref{coro:orientation:EFX00_EFX0-})} & \makecell{always exist, P\\ (Theorem \ref{thm:allocation:EFX0-})} \\
    \cline{2-4}
    &$\EFX^+_0$ & \makecell{may not exist, NP-c\\ (Corollary \ref{coro:orientation:EFX+0})} & \makecell{always exist, P\\ (Theorem \ref{thm:allocation:EFX+0})} \\
    \cline{2-4}
    &$\EFX^+_-$ & \makecell{may not exist, NP-c\\ (Theorem \ref{thm:orientation:EFX+-})} & \makecell{always exist, P\\ (Corollary \ref{coro:allocation:EFX+-})} \\
    \hline
\end{tabular}
\caption{A summary of results. ``NP-c'' means that determining the existence of the corresponding orientations/allocations is NP-complete. ``P'' means that the allocations can be found in polynomial time if they always exist, or determining the existence is in polynomial time if they may not exist.}
\label{tb:our_results}
\end{table}

In this work, we provide a complete study of the model in \citep{christodoulou2023fair} by considering the allocations (and orientations) of goods, chores, and mixed manna. 
For the case of goods, we consider the two variants of EFX in the literature, i.e., $\EFX^0$ and $\EFX^+$, where the super script $0$ (resp. $+$) indicates that the item virtually removed from the envied agent's bundle has a non-negative (resp. strictly positive) marginal value to the envious agent. 
While most results for goods have been completed or can be implied by those in \citep{christodoulou2023fair}, we prove that an $\EFX^+$ orientation always exists and can be computed in polynomial time, in contrast to $\EFX^0$ orientations which may not exist. 

For the case of chores, we also consider two variants of EFX, i.e., $\EFX_0$ and $\EFX_-$, where the sub script $0$ (resp. $-$) indicates that the item virtually removed from the envious agent's bundle has a non-positive (resp. strictly negative) marginal value to the envious agent. 
We show that while an envy-free (thus is $\EFX_0$ and $\EFX_-$) allocation always exists and can be computed in polynomial time, an $\EFX_0$ or $\EFX_-$ orientation may not exist but their existence can be determined in polynomial time. 
A recent independent work by \cite{hsu2025polynomial} also proves that determining the existence of $\EFX_0$ orientations is in polynomial time, using a different approach from ours. 
We note that although the valuations they defined are slightly more general than ours, our approach could carry over to their problem. 

For the case of mixed manna, we consider the four variants of EFX in the literature, i.e., $\EFX^0_0$, $\EFX^0_-$, $\EFX^+_0$, $\EFX^+_-$, where the super and sub scripts have the same meanings as aforementioned.  
We first show that an orientation that satisfies any of the four EFX notions may not exist, and determining its existence is NP-complete.
While the NP-completeness of $\EFX^0_0$ and $\EFX^0_-$ orientations can be implied by the result in \citep{christodoulou2023fair}, it is worth noting that our approach for the NP-completeness of $\EFX^0_0$ allocations provides an alternative proof. 
For the NP-completeness of $\EFX^+_0$ and $\EFX^+_-$ orientations, we present a novel reduction from the $(3, B2)$-SAT problem, which depends on a delicate structure with chore edges. 
Due to the hardness results for orientations, we turn to study some specific graphs such as trees, stars and paths. 
For these graphs, we show that $\EFX^+_0$ or $\EFX^+_-$ orientations always exist and can be computed in polynomial time. 
Although $\EFX^0_0$ or $\EFX^0_-$ orientations may not exist, determining their existence is in polynomial time. 
We then study the setting where the edges can be allocated to any agent. 
We show that an $\EFX^0_0$ allocation may not exist and determining its existence is NP-complete. 
Our reduction borrows the idea from the reduction by \citet{christodoulou2023fair} (see Theorem 2 in their paper) and generalizes it. 
Our reduction has a slightly more complex structure than theirs to accommodate the allocation model. 
As a result, our reduction can yield their result while theirs cannot carry over to our problem. 
We also show that an allocation that satisfies any of the other three notions always exists and can be computed in polynomial time. 
Our algorithm that computes $\EFX^0_-$ allocations is based on the algorithm in \citet{christodoulou2023fair}. 
However, since there is one more requirement in our problem that agents cannot envy others after removing a chore from their own bundles, we need to carefully allocate the edges that are chores to their endpoint agents. 
Besides, our algorithm relies on some properties that are not revealed or may not be satisfied by the algorithm in \citet{christodoulou2023fair}. 
A summary of our results can be seen in Table \ref{tb:our_results}. 

\subsection{Subsequent Work} 
There have been plenty of subsequent works since \citep{christodoulou2023fair} and the conference version of this paper \citep{DBLP:conf/ijcai/0027WL024}. 
A line of them extend the model to multigraphs where a pair of agents may be interested in more than one common item.  
\cite{afshinmehr2024efx} and \cite{bhaskar2024efx} considered some specific multigraphs (like bipartite multigraphs, multicycles, multitrees) and showed the existence of EFX allocations. 
For general multigraphs, \cite{amanatidis2024pushing} proved that 2/3-approximate EFX allocations always exist and can be computed in polynomial time. 
EFX orientations have also been studied for multigraphs. 
\cite{afshinmehr2024efx} and \cite{deligkas2024ef1} proved the NP-completeness of determining the existence of EFX orientations in some specific multigraphs. 
Besides, \cite{deligkas2024ef1}, \cite{hsu2024efx} and \cite{zeng2024structure} provided some characterization of multigraphs for which EFX orientations exist. 
Other extensions have also been studied. 
\cite{deligkas2024ef1} proposed the generalized setting where each agent has a predetermined set of items in which they are interested and showed that EF1 orientations always exist.  
\cite{kaviani2024almost} parameterized the number of common items which a pair of agents may value positively and the number of agents to whom an item may have non-zero values and proposed the so-called $(p, q)$-bounded valuations. 
\cite{chandramouleeswaran2024fair} studied EF1 orientations in a variable setting where an item may get lost or a new agent may arrive. 
Besides fairness, \cite{misra2024envy} also considered efficiency under \citep{christodoulou2023fair}'s model. 
They studied price of envy-freeness and the complexity of computing envy-free and efficient allocations. 
Besides indivisible goods, some work also studied \citep{christodoulou2023fair}'s model for chores. \cite{hsu2025polynomial} proved that determining the existence of EFX orientations for simple graphs is in polynomial time. 
They also proved that when considering multigraphs, the problem becomes NP-complete.

\subsection{Other Related Work}

Since \cite{Steinhaus49}, the theory of fair allocation has been extensively studied.
Traditional research mostly centers around allocating divisible resources, also known as the cake-cutting problem \cite{books/daglib/0017734,Procaccia_cake_16,LindnerR16}. 
In contrast to the divisible setting, the allocation of indivisible items becomes intractable, where the ideal fairness notions, including envy-freeness \cite{GS58,Varian74} and proportionality \cite{Steinhaus49}, are not guaranteed to be satisfiable. 
This leads to several research problems. 
For example, it is necessary to decide whether an instance admits a fair allocation, and to what extent the relaxation of these fairness notions can be satisfied. 
To answer the former question, it is easy to observe that deciding whether an instance admits an envy-free (or proportional) allocation is NP-hard, via a polynomial-time reduction from the Partition problem.s
In fact, it is proved by \cite{aziz2015fair,hosseini2020fair} that even if all items have binary additive valuations, the problem remains NP-hard. 
To answer the second question, envy-free up to one or any item (EF1, EFX) \citep{lipton2004approximately,budish2011combinatorial,DBLP:journals/teco/CaragiannisKMPS19} and maximin share (MMS) fairness \citep{budish2011combinatorial,akrami2024breaking} have been proposed and studied. 
It is proved by \cite{bhaskar2021approximate} that an EF1 allocation is always satisfiable under doubly monotone valuations. 
The existence of EFX allocations is still less understood. 

In addition to the traditional studies mentioned above, our paper is closely related to three specific lines of research, which we discuss below.

\paragraph{Constrained Fair Allocation}
Restricting who can get what is a well-motivated variant of the fair allocation problem.
Various types of constraints have been studied. 
For example, \cite{DBLP:conf/ijcai/BiswasB18,DBLP:conf/eumas/HummelH22} studied the cardinality constraints which restrict the number of items each agent can receive. 
A more general version of cardinality constraints is budget-feasible, where each item is at a price and the total payment of each agent is limited \cite{DBLP:journals/iandc/WuLG25,DBLP:conf/sigecom/Barman0SS23,DBLP:conf/sagt/ElkindIT24}.
Additionally, \cite{DBLP:journals/teco/LiV21,DBLP:journals/corr/abs-2404-11582} studied the hereditary constraints, and \cite{DBLP:conf/nips/LiLZ21,DBLP:conf/atal/KumarEGNV24} studied the scheduling constraints.
It is worth noting that certain constraints can be integrated into the valuations.  
For instance, cardinality, scheduling, and hereditary constraints are indeed special cases of XoS valuations that have been studied in \cite{DBLP:journals/ai/GhodsiHSSY22,DBLP:conf/nips/AkramiMSS23,DBLP:journals/ai/SeddighinS24}. 
Similar to our model, there are other constraints defined using graphs, which we will discuss later. 
We refer the readers to the survey by \citet{suksompong2021constraints} for a more comprehensive introduction to constrained fair allocation.

\paragraph{Mixed Manna}
Since initiated by \cite{DBLP:journals/aamas/AzizCIW22}, there has been a rapidly growing interest in fair allocation of 
mixed manna in recent years. 
\cite{DBLP:journals/aamas/AzizCIW22} proposed a polynomial-time algorithm named double round-robin that computes an envy-free up to one item (EF1) allocation for any instance with additive valuations. 
\cite{bhaskar2021approximate} extended this result by computing EF1 allocations for instances with general valuations. 
Their algorithm is based on the envy-cycle procedures and also runs in polynomial time. 
\cite{aleksandrov2019greedy} gave a polynomial-time algorithm that computes an EFX and Pareto-Optimal (PO) allocation for instances with identical and additive valuations. 
In their definition of EFX, the removed item must not have a value of zero, which as we will see, is the weakest one among the notions of EFX that we study in this paper. 
Other fairness notions studied for mixed manna include proportionality up to one item (PROP1) \citep{aziz2020polynomial} and maximin share (MMS) \citep{KulkarniMT21}. 
We refer the readers to the survey by \cite{DBLP:journals/corr/abs-2306-09564} for more details on this topic.

\paragraph{Fair Allocation on Graphs}
There are many other works that study fair allocation of indivisible items on graphs, whose settings whereas, are quite different from ours and \cite{christodoulou2023fair}'s. 
\cite{DBLP:conf/ijcai/BouveretCEIP17} formalized the problem that there is an underlying graph whose vertices are indivisible items and each agent must receive a connected component of the graph. 
They considered several fairness notions such as proportionality, envy-freeness, maximin share, and gave hardness results for general graphs and polynomial-time algorithms for special graphs.  
Many following works investigated the same problem with different fairness notions or graph structures \citep{DBLP:journals/geb/BiloCFIMPVZ22, DBLP:journals/dam/Suksompong19, DBLP:conf/aaai/IgarashiP19, DBLP:journals/aamas/BouveretCL19}. 
\cite{DBLP:journals/siamdm/BeiILS22} considered the same model and quantified the loss of fairness when imposing the connectivity constraint, i.e., \textit{price of connectivity}. 
\cite{DBLP:conf/ijcai/Madathil23} studied a similar model where each agent must receive a compact bundle of items that are ``closely related''. 
Different from this line of works, \cite{DBLP:journals/aamas/HummelH22} used a graph to reflect conflicts between items. 
Each vertex on the graph is an item and each edge means that its two endpoint items have a conflict. 
They require that two items that have a conflict cannot be allocated to the same agent. 
In other words, the bundle allocated to each agent must be an independent set of the graph. 
\cite{DBLP:conf/atal/PayanSV23} studied fair allocation on graph where vertices are agents (as in our setting). 
The graph was used to relax fairness notions such that fairness only need to be satisfied for the endpoint agents of the edges.

\subsection{Roadmap}
The rest of this paper is structured as follows. 
In Section \ref{sec:preliminaries}, we formally define the problem and the relative fairness notions. 
In Section \ref{sec:goods_and_chores}, we briefly discuss the results for goods instances and then present the detailed results for chores instances. 
Sections \ref{sec:mixed:ori} and \ref{sec:mixed:alloc} present the results for mixed instances, on EFX orientations and EFX allocations, respectively. 
Appendix \ref{ap:allocation:EFX0-:part1} contains some missing materials in Section \ref{sec:mixed:alloc}. 

\section{Preliminaries}
\label{sec:preliminaries}
For any positive integer $k$, let $[k] = \{1, \ldots, k\}$. 
In an instance of our problem, there is a simple graph $G = (N, M)$ where $N = \{a_1, \ldots, a_n\}$ is the vertex set and $M$ is the edge set. 
Each vertex corresponds to an agent and each edge corresponds to an indivisible item. 
We use vertex and agent, edge and item, interchangeably. 
We also write both $(a_i, a_j)$ and $e_{i, j}$ to represent the edge between $a_i$ and $a_j$. 
An {\em allocation} $\mathbf{X} = (X_1, \ldots, X_n)$ is an $n$-partition of $M$ such that $X_i$ contains the edges allocated to agent $a_i$, where $X_i \cap X_j = \emptyset$ for any $a_i, a_j \in N$ and $\bigcup_{a_i \in N}X_i = M$. 
An {\em orientation} is a restricted allocation where each edge must be allocated to one of its endpoint agents; that is, for each edge $e_{i, j} \in M$, either $e_{i, j} \in X_i$ or $e_{i, j} \in X_j$. 
An allocation $\mathbf{X}$ is partial if $\bigcup_{a_i \in N}X_i \subsetneq M$. 

Each agent $a_i \in N$ has a valuation $v_i: 2^M\to {\mathbb R}$ over the edges  such that the items $M$ are divided into three groups: \textit{goods} with $v_i(S \cup \{e\}) > v_i(S)$ for any subset $S \subseteq M \setminus \{e\}$, \textit{chores} with $v_i(S \cup \{e\}) < v_i(S)$ for any subset $S \subseteq M \setminus \{e\}$, and \textit{dummies} with $v_i(S \cup \{e\}) = v_i(S)$ for any subset $S \subseteq M \setminus \{e\}$. 
An incident edge of $a_i$ is either a good, a chore or a dummy for $a_i$, and any edge that is not incident to $a_i$ must be a dummy for $a_i$. 
Note that $v_i$ is not necessarily additive and we consider a slightly adjusted variant of the doubly monotone valuations \citep{DBLP:journals/aamas/AzizCIW22} such that the strictness of the marginal value of every item is the same for every subset. 
Let $E_i$ be the set of all edges that are incident to $a_i$, and $E_i^{\ge 0}, E_i^{> 0}, E_i^{= 0}, E_i^{< 0}\subseteq E_i$ be the subsets of non-chores, goods, dummies and chores, respectively. 

An instance is called a \textit{goods instance} if every item is either a good or a dummy for every agent, a \textit{chores instance} if every item is either a chore or a dummy for every agent, or a \textit{mixed instance} if every item could be a good, a chore, or a dummy for every agent. 

\subsection{Fairness Notions}
Given an allocation $\mathbf{X}$ and two agents $a_i, a_j\in N$, we say that $a_i$ envies $a_j$ if $v_i(X_j) > v_i(X_i)$. 
The allocation is \textit{envy-free} (EF) if no agent envies any other agent, i.e., $v_i(X_i) \ge v_i(X_j)$ for any $a_i, a_j \in N$. 
As we have seen, envy-freeness is too demanding for indivisible items. 
Thus, in this paper, we focus on its prominent relaxations \textit{envy-free up to any item} (EFX). 

For goods instances, there are two variants of EFX in the literature, i.e., $\EFX^0$ \citep{plaut2020almost} and $\EFX^+$ \citep{DBLP:journals/teco/CaragiannisKMPS19}, which require that any envy could be eliminated after virtually removing any item from the envied agent's bundle that has a non-negative or a strictly positive marginal value to the envious agent, respectively. 
The formal definitions are as below, 
\begin{definition}[$\EFX^0$]
    An allocation $\mathbf{X} = (X_1, \ldots, X_n)$ is $\EFX^0$ if for every two agents $a_i, a_j \in N$ such that $a_i$ envies $a_j$ and any $e \in X_j$ such that $v_i(X_j \setminus \{e\}) \le v_i(X_j)$, 
    we have $v_i(X_i) \ge v_i(X_j \setminus \{e\})$. 
\end{definition}

\begin{definition}[$\EFX^+$]
    An allocation $\mathbf{X} = (X_1, \ldots, X_n)$ is $\EFX^+$ if for every two agents $a_i, a_j \in N$ such that $a_i$ envies $a_j$ and any $e \in X_j$ such that $v_i(X_j \setminus \{e\}) < v_i(X_j)$, we have $v_i(X_i) \ge v_i(X_j \setminus \{e\})$. 
\end{definition}

Accordingly, for chores instances, we also consider two variants of EFX, i.e., $\EFX_0$ and $\EFX_-$, which require that any envy could be eliminated after virtually removing any item from the envious agent's own bundle that has a non-positive or a strictly negative marginal value to herself, respectively. 
The formal definitions are as below, 
\begin{definition}[$\EFX_0$]
    An allocation $\mathbf{X} = (X_1, \ldots, X_n)$ is $\EFX_0$ if for every two agents $a_i, a_j \in N$ such that $a_i$ envies $a_j$ and any $e \in X_i$ such that $v_i(X_i \setminus \{e\}) \ge v_i(X_i)$, we have $v_i(X_i \setminus \{e\}) \ge v_i(X_j)$. 
\end{definition}

\begin{definition}[$\EFX_-$]
    An allocation $\mathbf{X} = (X_1, \ldots, X_n)$ is $\EFX_-$ if for every two agents $a_i, a_j \in N$ such that $a_i$ envies $a_j$ and any $e \in X_i$ such that $v_i(X_i \setminus \{e\}) > v_i(X_i)$, we have $v_i(X_i \setminus \{e\}) \ge v_i(X_j)$. 
\end{definition}

For mixed instances, there are four variants of EFX in the literature \citep{DBLP:journals/aamas/AzizCIW22, DBLP:conf/ki/AleksandrovW20, DBLP:journals/corr/abs-2006-04428}, namely, $\EFX_0^0$, $\EFX^0_-$, $\EFX^+_0$, and $\EFX^+_-$. 
$\EFX_0^0$ requires that any envy could be eliminated by removing any item that is not a chore for the envious agent from the envied agent's bundle or any item that is not a good from the envious agent's own bundle. 
Formally, 

\begin{definition}[$\EFX^0_0$]
    An allocation $\mathbf{X} = (X_1, \ldots, X_n)$ is $\EFX^0_0$ if for every two agents $a_i, a_j \in N$ such that $a_i$ envies $a_j$, both of the following conditions hold: 
    \begin{enumerate}
        \item for any $e \in X_j$ such that $v_i(X_j \setminus \{e\}) \le v_i(X_j)$, $v_i(X_i) \ge v_i(X_j \setminus \{e\})$;
        \item for any $e \in X_i$ such that $v_i(X_i \setminus \{e\}) \ge v_i(X_i)$, $v_i(X_i \setminus \{e\}) \ge v_i(X_j)$. 
    \end{enumerate}
\end{definition}

$\EFX^0_-$ differs from $\EFX_0^0$ in that the item removed from the envious agent's bundle cannot be a dummy. 
More concretely, the item $e$ considered in the second condition is subject to $v_i(X_i \setminus \{e\}) > v_i(X_i)$. 
$\EFX^+_0$ differs from $\EFX_0^0$ in that the item removed from the envied agent's bundle cannot be a dummy, i.e., the item $e$ considered in the first condition is subject to $v_i(X_j \setminus \{e\}) < v_i(X_j)$. 
$\EFX^+_-$ differs from $\EFX_0^0$ in that the item removed from the envied agent's and the envious agent's bundles cannot be a dummy. 
Formally,  

\begin{definition}[$\EFX^0_-$]
    An allocation $\mathbf{X} = (X_1, \ldots, X_n)$ is $\EFX^0_-$ if for every two agents $a_i, a_j \in N$ such that $a_i$ envies $a_j$, both of the following conditions hold: 
    \begin{itemize}
        \item for any $e \in X_j$ s.t. $v_i(X_j \setminus \{e\}) \le v_i(X_j)$, $v_i(X_i) \ge v_i(X_j \setminus \{e\})$; 
        \item for any $e \in X_i$ s.t. $v_i(X_i \setminus \{e\}) > v_i(X_i)$, $v_i(X_i \setminus \{e\}) \ge v_i(X_j)$. 
    \end{itemize}
\end{definition}

\begin{definition}[$\EFX^+_0$]
    An allocation $\mathbf{X} = (X_1, \ldots, X_n)$ is $\EFX^+_0$ if for every two agents $a_i, a_j \in N$ such that $a_i$ envies $a_j$, both of the following conditions hold: 
    \begin{itemize}
        \item for any $e \in X_j$ s.t. $v_i(X_j \setminus \{e\}) < v_i(X_j)$, $v_i(X_i) \ge v_i(X_j \setminus \{e\})$; 
        \item for any $e \in X_i$ s.t. $v_i(X_i \setminus \{e\}) \ge v_i(X_i)$, $v_i(X_i \setminus \{e\}) \ge v_i(X_j)$. 
    \end{itemize}
\end{definition}

\begin{definition}[$\EFX^+_-$]
    An allocation $\mathbf{X} = (X_1, \ldots, X_n)$ is $\EFX^+_-$ if for every two agents $a_i, a_j \in N$ such that $a_i$ envies $a_j$, both of the following conditions hold: 
    \begin{itemize}
        \item for any $e \in X_j$ s.t. $v_i(X_j \setminus \{e\}) < v_i(X_j)$, $v_i(X_i) \ge v_i(X_j \setminus \{e\})$; 
        \item for any $e \in X_i$ s.t. $v_i(X_i \setminus \{e\}) > v_i(X_i)$, $v_i(X_i \setminus \{e\}) \ge v_i(X_j)$. 
    \end{itemize}
\end{definition}

\section{Goods and Chores Instances}
\label{sec:goods_and_chores}
In this section, we first briefly introduce the results for goods instances and then focus on the detailed results for chores instances. 

\subsection{Goods Instances}
\label{ap:goods} 
For goods instances, \cite{christodoulou2023fair} have showed that an $\EFX^0$ allocation always exists and can be computed in polynomial time, which also holds for $\EFX^+$ allocations since $\EFX^0$ implies $\EFX^+$. 

\begin{corollary}\label{cor:goods_orien:EFX+}
    For any goods instance, an $\EFX^+$ allocation always exist and can be computed in polynomial time. 
\end{corollary}

They also showed that an $\EFX^0$ orientation may not exist and determining its existence is NP-complete. 
Our reduction in Theorem \ref{thm:allocation:EFX00} provides an alternative proof to this result.
Contrary to $\EFX^0$, it is easy to compute $\EFX^+$ orientations.

\begin{proposition}
\label{prop:goods_orien:EFX+}
    For any goods instance, an $\EFX^+$ orientation always exist and can be computed in polynomial time. 
\end{proposition}
\begin{proof}
     Our algorithm is to arbitrarily allocate each edge to one of its endpoint agents. 
     Consider any two agents $a_i$, $a_j$ such that $a_i$ envies $a_j$. 
     It is clear that $a_j$ must receive $e_{i, j}$ and $e_{i, j}$ is the only edge in $a_j$'s bundle that has a strictly positive marginal value to $a_i$. 
     After removing $e_{i, j}$ from $a_j$'s bundle, $a_i$ values $a_j$'s bundle at 0 and does not envy $a_j$. 
\end{proof}

\subsection{Chores Instances}
\label{ap:chores}
For chores instances, an envy-free allocation can be easily computed by arbitrarily allocating each edge to an agent who is not its endpoint, which is also $\EFX_0$ and $\EFX_-$.

\begin{proposition}\label{prop:chores_allo}
    For any chores instance, an envy-free (thus $\EFX_0$ and $\EFX_-$) allocation always exist and can be computed in polynomial time. 
\end{proposition}

However, an $\EFX_0$ or $\EFX_-$ orientation may not exist. 
\begin{example}
    Consider a connected graph where there are more edges than vertices and each edge is a chore for both its two endpoint agents. 
    By the pigeonhole principle, there must exist an agent who receives more than one edge. 
    After removing one edge from her own bundle, that agent still receives a negative value and envies the other endpoint agents of the edges she receives since she values those agents' bundles at 0.  
\end{example}
The good news is that we can determine the existence of $\EFX_-$ and $\EFX_0$ orientations in polynomial time. 
Let us consider the simpler notion $\EFX_-$ first. 

\begin{proposition}\label{prop:chores_orien:EFX-}
    For any chores instance, whether an $EFX_-$ orientation exists or not can be determined in polynomial time. 
\end{proposition}
\begin{proof}
    Given any chores instance on a graph $G = (N, M)$, we construct a new graph $G^\prime = (N, M^\prime)$ by deleting all edges in $M$ that are dummies for at least one of their endpoint agents. 
    We next show that there exists an $\EFX_-$ orientation on $G$ if and only if there exists one on $G^\prime$. 
    On one hand, if there exists an $\EFX_-$ orientation $X$ on $G$, we construct an orientation $X^\prime$ on $G^\prime$ by removing the edges in $M \setminus M^\prime$ from $X$. 
    Clearly, $X^\prime$ is still $\EFX_-$. 
    On the other hand, if there exists an $\EFX_-$ orientation $X^\prime$ on $G^\prime$, we construct an orientation $X$ on $G$ by allocating each edge in $M \setminus M^\prime$ to an endpoint agent for whom the edge is a dummy. 
    Since no agent gets worse off and the added edges are not considered in the definition of $\EFX_-$, $X$ is still $\EFX_-$. 

    Then it suffices to show how to determine the existence of an $\EFX_-$ orientation on $G^\prime$. 
    Note that $G^\prime$ consists of many connected components and each edge in $G^\prime$ is a chore for both its endpoint agents. 
    For each connected component, we compare the number of edges and vertices. 
    If there exists one connected component such that the number of edges is larger than that of vertices, we conclude that no orientation is $\EFX_-$ on $G^\prime$. 
    This is because by the pigeonhole principle, there must exist an agent in that connected component who receives more than one edge. 
    After removing one edge from her own bundle, that agent still receives a negative value and envies the other endpoint agents of the edges she receives since she values those agents' bundles at 0.  
    If no such connected component exists, we can compute an $\EFX_-$ orientation by allocating each agent exactly one edge. 
\end{proof}

In the following, we elaborate on how to determine the existence of $\EFX_0$ orientations for chores instances in polynomial time. 
We have the following main theorem. 

\begin{theorem}\label{thm:chores_orien:EFX0}
    For any chores instance, we can determine whether an $\EFX_0$ orientation exists or not in polynomial time.
\end{theorem}

We start with a characterization of $\EFX_0$ orientations for chores instances.
\begin{lemma}\label{lem:orientation:EFX0:characterization}
    For chores instances, an orientation $\mathbf{X} = (X_1, \ldots, X_n)$ is $\EFX_0$ if and only if each agent either does not receive any edge that is a chore for them, or receives exactly one edge which is a chore. 
    That is, for any agent $a_i \in N$, 
    \begin{enumerate}
        \item if there exists an edge $e \in X_i$ such that $e \in E_i^{<0}$, then $|X_i| = 1$,
        \item otherwise for every edge $e \in X_i, e \in E_i^{=0}$.
    \end{enumerate}
\end{lemma}
\begin{proof}
    Suppose that an orientation $X$ satisfies the above conditions. 
    If an agent satisfies the first condition, then after removing the only edge in her bundle, she does not envy others. 
    Otherwise (i.e., the agent satisfies the second condition), she receives a value of 0 and does not envy others. 
    Thus, $X$ is $\EFX_0$. 
    Now we show that an $\EFX_0$ orientation must satisfy the conditions. 
    For a fixed agent $a_i$, if no edge in her bundle is a chore for her, then the second condition holds. 
    Otherwise, let $e_{i, j}$ be one edge in $a_i$'s bundle that is a chore for $a_i$. 
    If there exists another edge $e^\prime \neq e_{i, j}$ in $a_i$'s bundle, then the orientation is not $\EFX_0$, a contradiction. 
    This is because $v_i(X_i\setminus \{e^\prime\}) \leq v_i(e_{i, j}) < 0$ and $v_i(X_j) = 0$ since all edges in $X_j$ are not incident to $a_i$ and are thus dummies for $a_i$. 
\end{proof}
Note that the if and only if condition in Lemma \ref{lem:orientation:EFX0:characterization} is a local condition for each agent $a_i\in N$. 
Therefore, we may say in the following proof that an agent $a_i \in N$ (or an agent set $U \subseteq N$) satisfies the condition in Lemma \ref{lem:orientation:EFX0:characterization}.

We next define two important properties for partial orientations.

\begin{definition}[Implementable]\label{def:orientation:EFX0:feasible}
    For chores instances, we say that a partial orientation $\mathbf{X}^p = (X_1^p, \ldots, X_n^p)$ is implementable (for $\EFX_0$) if there exists an $\EFX_0$ orientation $\mathbf{X}=(X_1, \ldots, X_n)$ such that $X_i^p \subseteq X_i$ for every $a_i \in N$. We call such $\mathbf{X}$ an implementation of $\mathbf{X}^p$.
\end{definition}

\begin{definition}[Maximal]\label{def:orientation:EFX0:maximal}
    For chores instances, we say that a partial orientation $\mathbf{X}^p = (X_1^p$ $, \ldots, X_n^p)$ is maximal if for some subset of agents $N^\prime$, all adjacent edges of $N^\prime$ are allocated to $N^\prime$. That is, $\bigcup_{a_i\in N^\prime} X_i^p = \bigcup_{a_i\in N^\prime} E_i$.
\end{definition}

The following lemma shows that some special maximal orientations are implementable. This lemma will be a basic ingredient to the polynomial-time algorithm that computes a potentially existing $\EFX_0$ orientation.

\begin{lemma}\label{lem:orientation:EFX0:subproblem}
    For a chores instance such that an $\EFX_0$ orientation exists, let $\mathbf{X}^p = (X_1^p, \ldots, X_n^p)$ be a maximal (partial) orientation with respect to a set of agents $N^\prime \subset N$ such that every agent in $N \setminus N^\prime$ receives an empty bundle. 
    If $\mathbf{X}^p$ satisfies the condition in Lemma \ref{lem:orientation:EFX0:characterization}, then it is implementable. 
\end{lemma}

\begin{proof}
    Let $\mathbf{X} = (X_1, \ldots, X_n)$ be an $\EFX_0$ orientation. 
    We construct an implementation $\mathbf{X}^\prime = (X_1^\prime, \ldots, X_n^\prime)$ of $\mathbf{X}^p$ by reallocating to $N^\prime$ the edges that are incident to $N^\prime$ but are allocated to $N \setminus N^\prime$ in $\mathbf{X}$; 
    That is, $X_i^\prime = X_i^p$ for every $a_i \in N^\prime$ and $X_i^\prime = X_i \setminus \bigcup_{a_i \in N^\prime} E_i$. 
    We next show that $\mathbf{X}^\prime$ satisfies the condition in Lemma \ref{lem:orientation:EFX0:characterization} and thus is $\EFX_0$. 
    For every $a_i \in N^\prime$, the condition holds because of the assumption. 
    For every $a_i \in N \setminus N^\prime$, since $\mathbf{X}$ is an $\EFX_0$ orientation, we know $X_i$ satisfies the condition in Lemma \ref{lem:orientation:EFX0:characterization}. Together with the fact that $X_i^\prime \subseteq X_i$, we claim that $X_i^\prime$ also satisfies the condition in Lemma \ref{lem:orientation:EFX0:characterization}. Therefore, $\mathbf{X}^\prime$ satisfies the condition in Lemma \ref{lem:orientation:EFX0:characterization} and thus is an $\EFX_0$ orientation.
\end{proof}

\begin{algorithm}[tb]
\caption{Determining the Existence of an $\EFX_0$ Orientation}
\label{alg:orientation:EFX0}
\KwIn{A chores instance with a graph where $N$ is the vertex set and $M$ is the edge set.}
\KwOut{A boolean $flag$ indicating whether an $\EFX_0$ orientation exists or not, and an $\EFX_0$ orientation $\mathbf{X} = (X_1, \ldots, X_n)$ if $flag$ is true.}
Initialize $X_i \leftarrow \emptyset$ for every $a_i \in N$, $R \leftarrow M$. \\
\While{$R \neq \emptyset$} {
    Let $a_i\in N$ be an arbitrary agent with $R \cap E_i \neq \emptyset$ and initialize $flag \gets false$. \\
    \For{each edge $e \in R \cap E_i$}{
        $(flag, \mathbf{X^{\prime}}, U, R) \gets \mathbf{Push}(\mathbf{X}, e, a_i, R).$ \\
        \If{$flag = true$}{
            $\mathbf{X} \gets \mathbf{X^{\prime}}.$ \\
            $R \gets R \setminus \bigcup_{a_j\in U} E_j$. \\
            break.
        }
    }
    \If{$flag = false$}{
        \Return $(false, null)$.
    }
}
\Return $(true, \mathbf{X})$.
\end{algorithm}

\begin{algorithm}[tb]
\caption{\textbf{Push} Subroutine for Algorithm \ref{alg:orientation:EFX0}}
\label{alg:orientation:EFX0:Push}
\KwIn{A (partial) orientation $\mathbf{X} = (X_1, \ldots, X_n)$, an edge $e$, an agent $a_u$ to whom to push $e$, and the set $R$ of unallocated edges.}
\KwOut{A boolean $flag$, and if $flag$ is true, a (partial) orientation $\mathbf{X^{\prime}} = (X_1^{\prime}, \ldots, X_n^{\prime})$, a set of agents $U$ whose bundles are updated, and the set $R$ of unallocated edges.}
\If{$X_u \neq \emptyset$}{
    \Return $(false, null, null, null)$.
}
Initialize $X_i^{\prime} \leftarrow X_i$ for every $a_i \in N$ and $U\gets \{a_u\}$. \\
$X_u^\prime \gets X_u^\prime \cup \{e\}$, $R \leftarrow R \setminus\{e\}$. \\
\If{$e$ is a dummy to $a_u$}{
    $X_u^\prime \gets X_u^\prime \cup (E_u^{=0} \cap R)$. \\
    $R \leftarrow R \setminus E_u^{=0}$. 
}
\For{each $e_{u, j} \in (E_u \cap R) \setminus X_u^\prime$}{
    $(flag, \mathbf{X}^\prime, U^\prime, R) \gets \mathbf{Push}(\mathbf{X}^\prime, e_{u, j}, a_j, R)$. \\
    \If{$flag = false$}{
        \Return $(false, null, null, null)$.
    }
    $U \gets U \cup U^\prime$.
}
\Return $(true, \mathbf{X}^\prime, U, R)$.
\end{algorithm}

Now we are ready to introduce our algorithms (i.e., Algorithm \ref{alg:orientation:EFX0} and its subroutine Algorithm \ref{alg:orientation:EFX0:Push}). 
Algorithm \ref{alg:orientation:EFX0} primarily consists of a while-loop that each time allocates all adjacent edges of some agents to those agents. 
In the loop body, the algorithm first picks an arbitrary agent $a_i$ who has some unallocated adjacent edges. Then, it calls Algorithm \ref{alg:orientation:EFX0:Push} to check that whether it is feasible to push some adjacent edge $e$ to $a_i$. If yes, the algorithm updates the (partial) orientation according to the return of Algorithm \ref{alg:orientation:EFX0:Push}. Otherwise, the algorithm returns that no orientation is $\EFX_0$. 

The subroutine Algorithm \ref{alg:orientation:EFX0:Push} is a recursive process. 
It first checks whether $a_u$ has received any edge or not. 
If yes, either $a_u$ has received an edge that is a chore for her, or she has received some edges that are all dummies for her and $e$ is a chore for her. 
In both cases, $a_u$ cannot receive $e$ any more which will break the condition in Lemma \ref{lem:orientation:EFX0:characterization}, and thus Algorithm \ref{alg:orientation:EFX0:Push} returns $false$. 
If $a_u$ has not received any edge yet, the algorithm allocates $e$ to her, and if $e$ is a dummy for her, also allocates to her all other unallocated adjacent edges that are dummies for her. Finally, it recursively pushes each other unallocated adjacent edge of $a_u$ to the other endpoint agent of the edge. As long as the algorithm returns $flag=false$ at some recursive call, it will immediately terminate and return $flag=false$ to the outer call.

We first prove the following claim to show how the recursive subroutine works. 

\begin{claim}\label{clm:orientation:EFX0:one_dfs_true}
    For Algorithm \ref{alg:orientation:EFX0:Push}, when $flag$ is $true$, the algorithm returns a (partial) maximal orientation, and each agent not in $U$ does not receive any edge in the process. 
\end{claim}

\begin{proof}
    Notice that Algorithm \ref{alg:orientation:EFX0:Push} is a recursive DFS process. We analyze from the view of the outer call from Algorithm \ref{alg:orientation:EFX0}. Since we assume the return $flag = true$, and anytime a recursive call to Algorithm \ref{alg:orientation:EFX0:Push} returns $flag = false$, the algorithm will terminate recursively and return $flag = false$ to the outer call, we may ignore the first 3 lines in Algorithm \ref{alg:orientation:EFX0:Push}. Then, the algorithm can be simplified to the following process.
    \begin{enumerate}
        \item Allocate $e$ to $a_u$, and if $e$ is a dummy for $a_u$, also allocate to her all other unallocated adjacent edges that are dummies for her.
        \item Recursively push each of the remaining adjacent edges of $a_u$ to the other endpoint agents of the edges.
        \item Record all pushed agents in the set $U$.
    \end{enumerate}
    We can see that for every agent $a_u$ in the output set $U$, all her adjacent edges are either allocated to $a_u$ or allocated to the other vertex of the edge by recursively push. Therefore, if the algorithm terminates with $flag = true$, for the output set $U$, $\bigcup_{i\in U} E_i \subseteq \bigcup_{i\in U} X_i$. Since player $i$ will only be allocated with her adjacent edge, this actually means that $\bigcup_{i\in U} E_i = \bigcup_{i\in U} X_i$. Therefore, the algorithm returns a maximal orientation with respect to $U$. On the other side, for every vertex not in $U$, the algorithm does not touch them so they are not updated.
\end{proof}

Moreover, we have the following corollary that guarantees the input of Algorithm \ref{alg:orientation:EFX0:Push} satisfies the requirement.

\begin{corollary}\label{cor:orientation:EFX0:dfs_input}
    Every time when Algorithm \ref{alg:orientation:EFX0} calls Algorithm \ref{alg:orientation:EFX0:Push}, the input partial orientation is guaranteed to be maximal with respect to some set $T$, and each vertex other than $T$ is allocated with no edge in the input orientation.
\end{corollary}

\begin{proof}
    This corollary follows directly from Claim \ref{clm:orientation:EFX0:one_dfs_true} and the way Algorithm \ref{alg:orientation:EFX0} updates the orientation $\mathbf{X}$. At the first time, Algorithm \ref{alg:orientation:EFX0} calls Algorithm \ref{alg:orientation:EFX0:Push} with an empty orientation. After the $t$-th call to Algorithm \ref{alg:orientation:EFX0:Push}, if the call returns $flag_t=false$, the algorithm will not update the orientation. Otherwise, the orientation is updated by some way to allocate all adjacent edges of $U_t$ to the vertex set $U_t$. Then, the new orientation allocates all adjacent edges of $\bigcup_{i\in[t], flag_i = true} U_i$ to the vertex set $\bigcup_{i\in[t], flag_i = true} U_i$, and for other vertex in $N \setminus \bigcup_{i\in[t], flag_i = true} U_i$, they are not allocated with any edge.
\end{proof}

The above actually corollary characterize the behavior of our algorithm when an $\EFX_0$ orientation does exist for the input instance. To see it more clearly, we turn to prove the main theorem for chores instances.

\begin{proof}[Proof of Theorem \ref{thm:chores_orien:EFX0}]
    We show the following three properties for the first time Algorithm \ref{alg:orientation:EFX0} runs the loop body, i.e., Line 3 to Line 14.
    \begin{enumerate}
        \item If the algorithm returns $flag=false$ at Line 13, there does not exist an $\EFX_0$ orientation for the input instance.
        \item Otherwise, the algorithm allocate all edges in $\bigcup_{i\in U} E_i$ to $U$ and, for any $i \in U$, the condition in Lemma \ref{lem:orientation:EFX0:characterization} is satisfied.
        \item This one run of the loop body uses polynomial time.
    \end{enumerate}
    If we are able to prove these three properties, together with Lemma \ref{lem:orientation:EFX0:subproblem},  we can use mathematical induction to prove the theorem. This is because for every run of the loop body except for the last, the algorithm calculate a maximal orientation with respect to some set $U$, and leaves the following run of the loop body a subproblem that consists of vertices $N \setminus U$ and edges $M \setminus \bigcup_{i\in U} E_i$. According to Lemma \ref{lem:orientation:EFX0:subproblem}, this orientation is implementable if there exists an $\EFX_0$ orientation for the original input instance. For the last run of the loop body, the algorithm either finds that there actually does not exist an $\EFX_0$ orientation for the subproblem, which means there does not exist an $\EFX_0$ orientation for the original problem, or completely solve the subproblem since $M = \emptyset$ after this run of loop body and returns an $\EFX_0$ orientation according to the characterization of Lemma \ref{lem:orientation:EFX0:characterization}. The polynomial run time is trivial since there are at most $n$ runs of the loop body and each of them runs in polynomial time.

    The following proof will focus on showing properties 1 to 3. For Property (1), since in Line 4 the algorithm enumerates all possible edge $e\in E_i$ that vertex $i$ can obtain, we only have to show that if in Line 5, Algorithm \ref{alg:orientation:EFX0:Push} returns $flag=false$, there does not exist an $\EFX_0$ allocation that assigns $e$ to $i$. Notice that Algorithm \ref{alg:orientation:EFX0:Push} is actually a necessary condition checker that finds if allocating $e$ to $i$, which other edges must be allocated to who such that the condition in Lemma \ref{lem:orientation:EFX0:characterization} is satisfied. Specifically, if $v_i(e) < 0$, agent $i$ cannot receive any other edge, so each remaining adjacent edge must be allocated to the other vertex of the edge. If $v_i(e) = 0$, agent $i$ can take any other adjacent edge $e'$ with $v_i(e') = 0$. The fact that allocating all such edge to $i$ is (weakly) optimal can be shown by an argument similar to the proof of Lemma \ref{lem:orientation:EFX0:subproblem}. That is, for any $\EFX_0$ orientation and some vertex $i$ that gets zero utility, there exists another $\EFX_0$ orientation that allocate all edge in $E_i^{=0}$ to vertex $i$, since other agents receive (weakly) less edges in this new orientation and condition in Lemma \ref{lem:orientation:EFX0:characterization} will not be broken. Therefore, the way the algorithm allocate edges is aligned with some $\EFX_0$ orientation, given that an $\EFX_0$ orientation exists. Notice that after this algorithm allocate some edge to some agent, this agent cannot receive any any more edges. Therefore, if at some time of the recursive call of Algorithm \ref{alg:orientation:EFX0:Push}, when it tries to allocate edge $e$ to $u$ and find that $u$ is already allocated with some edges, the algorithm finds the assumption that an $\EFX_0$ orientation exists is wrong, so it returns $flag=false$. 

    The other two properties are simpler. Property 2 directly follows from Claim \ref{clm:orientation:EFX0:one_dfs_true}. Properties 3 holds because one call to Algorithm \ref{alg:orientation:EFX0:Push} uses a time linear in the size of return $U$, and Algorithm \ref{alg:orientation:EFX0} calls Algorithm \ref{alg:orientation:EFX0:Push} $O(deg(i))$ times, where $i$ is the vertex selected in Line 3 and $deg(i)$ is the degree of vertex $i$. Therefore, one run of the loop body uses $O(|U|\max_{i\in N} deg(i))$ time. We remark that the total running time of Algorithm \ref{alg:orientation:EFX0} is $O(n \max_{i\in N} deg(i))$.
\end{proof}

\section{EFX Orientations for Mixed Instances}
\label{sec:mixed:ori}
In this section, we elaborate on EFX orientations for mixed instances. 
Firstly, we have the following proposition. 

\begin{proposition}\label{pro:orientation:neg_instance}
    There exist graphs for which no orientation satisfies any of the four notions of EFX. 
\end{proposition}
\begin{proof}
    Consider a connected graph where there are more edges than vertices and each edge is a chore for both its endpoint agents. 
    Since each edge must be allocated to one of its endpoint agents, there exists an agent who receives at least two chores in any orientation. 
    Even after removing one chore from her own bundle, this agent still receives a negative value and envies the agents who are not her neighbors. 
\end{proof}

Due to this negative result, we turn to studying the complexity of determining the existence of EFX orientations. 
The result by \cite{christodoulou2023fair} (see Theorem 2 in their paper) directly implies that determining the existence of $\EFX^0_-$ orientations is NP-complete. 
In the graphs constructed in their reduction, each edge is a good for both its endpoint agents. 
For such graphs, any $\EFX^0_-$ orientation is also $\EFX^0_0$.
Therefore, we have the following corollary. 

\begin{corollary}\label{coro:orientation:EFX00_EFX0-}
    Determining whether an $\EFX^0_0$ or $\EFX^0_-$ orientation exists or not is NP-complete. 
\end{corollary}

\subsection{$\EFX^+_-$ Orientations}
In the following, we prove the below theorem for $\EFX^+_-$. 

\begin{theorem}\label{thm:orientation:EFX+-}
    Determining whether an $\EFX^+_-$ orientation exists or not is NP-complete, even for additive valuations\footnote{Valuation $v_i$ is additive if $v_i(S) = \sum_{e \in S} v_i(e)$ for any $S \subseteq M$.}. 
\end{theorem}
To prove Theorem \ref{thm:orientation:EFX+-}, we reduce from $(3, B2)$-SAT problem to the $\EFX^+_-$ orientation problem. 
A $(3, B2)$-SAT instance contains a Boolean formula in conjunctive normal form consisting of $n$ variables $\{x_i\}_{i \in [n]}$ and $m$ clauses $\{C_j\}_{j \in [m]}$. 
Each variable appears exactly twice as a positive literal and exactly twice as a negative literal in the formula, and each clause contains three distinct literals. 
Determining whether a $(3, B2)$-SAT instance is satisfiable or not is NP-complete \citep{DBLP:journals/dam/BermanKS07}. 

Our reduction uses a gadget to ensure that a specific agent must receive a chore if the orientation is $\EFX^+_-$. 
One such gadget is shown in Figure \ref{fig:EFX+-:must_negative}. 
In this example, agent $a_i$ must receive $(a_i, a_1^{\Delta})$ if the orientation is $\EFX^+_-$. 
Otherwise, one of the other three agents must receive at least two chores and envy $a_i$ even after removing one chore. 

\begin{figure}[tb]
    \centering
    \includegraphics[width=0.23\columnwidth]{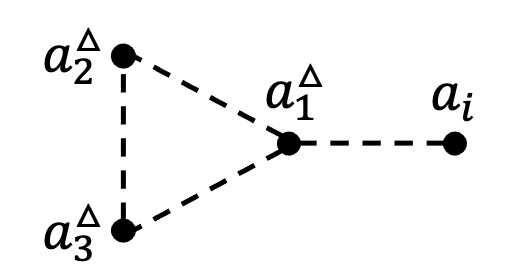}
    \caption{A gadget where agent $a_i$ must receive $(a_i, a_1^{\Delta})$ if the orientation is $\EFX^+_-$. Each dashed edge is a chore for both its endpoint agents.} 
    \label{fig:EFX+-:must_negative}
\end{figure}

Given a $(3, B2)$-SAT instance $(\{x_i\}_{i \in [n]}, \{C_j\}_{j \in [m]})$, we construct a graph as follows:
\begin{itemize}
    \item For each variable $x_i$, create two vertices $a^T_i, a^F_i$ and one edge $(a^T_i, a^F_i)$ with a value of 2 to both $a^T_i$ and $a^F_i$. 
    \item For each clause $C_j$, create one vertex $a_j^C$. 
    Besides, if $C_j$ contains a positive literal $x_i$, create one edge $(a_j^C, a^T_i)$ with a value of 1 to both  $a_j^C$ and $a^T_i$. 
    If $C_j$ contains a negative literal $\neg x_i$, create one edge $(a_j^C, a^F_i)$ with a value of 1 to both  $a_j^C$ and $a^F_i$. 
    \item Create three vertices $a_1^{\Delta}, a_2^{\Delta}, a_3^{\Delta}$ and three edges $(a_1^{\Delta}, a_2^{\Delta}), (a_2^{\Delta}, a_3^{\Delta}), (a_1^{\Delta}, a_3^{\Delta})$. 
    Besides, for each $i \in [n]$, create two edges $(a_i^T, a_{1}^{\Delta})$ and $(a_i^F, a_{1}^{\Delta})$. 
    For each $j \in [m]$, create one edge $(a_j^C, a_{1}^{\Delta})$. 
    Each of these edges has a value of $-1$ to both its endpoint agents. 
    \item Each vertex has an additive valuation.
\end{itemize}

To visualize the above reduction, we show the graph constructed from the formula $(x_1 \vee x_2 \vee x_3) \wedge (x_1 \vee x_2 \vee \neg x_3) \wedge (\neg x_1 \vee \neg x_2 \vee \neg x_3) \wedge (\neg x_1 \vee \neg x_2 \vee x_3)$ in Figure \ref{fig:EFX+-:reduction}. 

\begin{figure}[tb]
    \centering
    \includegraphics[width=0.48\columnwidth]{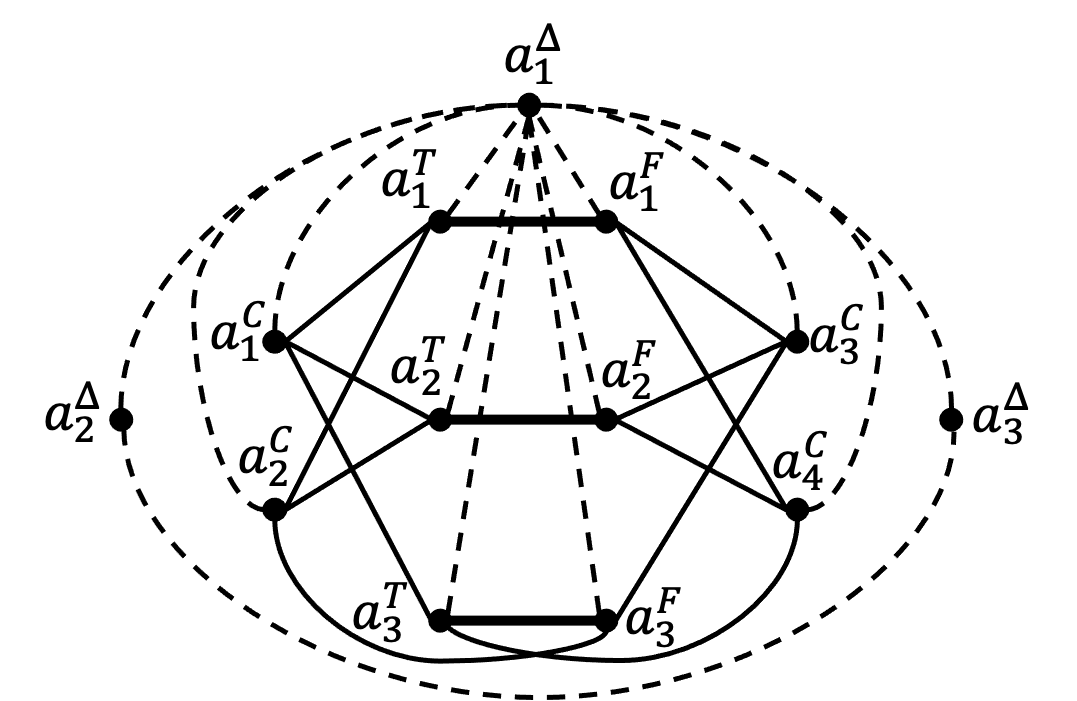}
    \caption{The graph constructed from the formula $(x_1 \vee x_2 \vee x_3) \wedge (x_1 \vee x_2 \vee \neg x_3) \wedge (\neg x_1 \vee \neg x_2 \vee \neg x_3) \wedge (\neg x_1 \vee \neg x_2 \vee x_3)$, where each edge has the same value to both its endpoint agents, each bold solid edge has a value of $2$, each non-bold solid edge has a value of $1$ and each dashed edge has a value of $-1$.} 
    \label{fig:EFX+-:reduction}
\end{figure}

We now prove that a $(3, B2)$-SAT instance is satisfiable if and only if the constructed graph has an $\EFX_-^+$ orientation. 

\begin{proof}[Proof of Theorem \ref{thm:orientation:EFX+-}]
    For ease of presentation, for each variable $x_i$, we denote by $C_{j^{i, T, 1}}, C_{j^{i, T, 2}}$ the two clauses that contain the positive literal $x_i$ and by $C_{j^{i, F, 1}}, C_{j^{i, F, 2}}$ the two clauses that contain the negative literal $\neg x_i$. 
    
    For one direction, we assume that the $(3, B2)$-SAT instance has a satisfying assignment and use the assignment to create an $\EFX_-^+$ orientation as follows: 
    \begin{itemize}
        \item Allocate $(a_1^{\Delta}, a_2^{\Delta})$ to $a_1^{\Delta}$, $(a_2^{\Delta}, a_3^{\Delta})$ to $a_2^{\Delta}$, and $(a_3^{\Delta}, a_1^{\Delta})$ to $a_3^{\Delta}$. 
        Allocate each other edge that is incident to $a_1^{\Delta}$ to the endpoint that is not $a_1^{\Delta}$. 
        \item For each variable $x_i$ that is set to True, allocate $(a_i^T, a_i^F)$ to $a_i^T$, 
        $(a_i^T, a_{j^{i, T, 1}}^C)$ to $a_{j^{i, T, 1}}^C$, $(a_i^T, a_{j^{i, T, 2}}^C)$ to $a_{j^{i, T, 2}}^C$, and $(a_i^F, a_{j^{i, F, 1}}^C), (a_i^F, a_{j^{i, F, 2}}^C)$ to $a_i^F$.  
        \item For each variable $x_i$ that is set to False, allocate $(a_i^T, a_i^F)$ to $a_i^F$, 
        $(a_i^F, a_{j^{i, F, 1}}^C)$ to $a_{j^{i, F, 1}}^C$, $(a_i^F, a_{j^{i, F, 2}}^C)$ to $a_{j^{i, F, 2}}^C$, and $(a_i^T, a_{j^{i, T, 1}}^C), (a_i^T, a_{j^{i, T, 2}}^C)$ to $a_i^T$.  
    \end{itemize}
    Next, we show that the above orientation is $\EFX_-^+$. 
    For agents $a_1^{\Delta}, a_2^{\Delta}, a_3^{\Delta}$, each of them receives one edge with a value of $-1$ and all edges received by other agents have non-positive values to them. 
    After removing the edge from their bundles, they do not envy others. 
    For each variable $x_i$ that is set to True, agent $a_i^T$ does not envy others since she receives a total value of 1 and each of her incident edges that she does not receive has a value of 1.  
    Agent $a_i^F$ receives three edges with values of $1, 1, -1$, respectively. 
    The only incident edge that she does not receive is $(a_i^T, a_i^F)$, which is allocated to $a_i^T$ and has a value of 2. 
    After removing the edge with a value of $-1$ from her own bundle or $(a_i^T, a_i^F)$ from $a_i^T$'s bundle, $a_i^F$ does not envy $a_i^T$. 
    We have an analogous argument for each variable that is set to False. 
    It remains to consider the agents that correspond to clauses. 
    Since the assignment is satisfying, each clause contains at least one literal that is evaluated to True. 
    This implies that each agent $a_j^C$ receives at least one edge with a value of $1$. 
    For example, if the clause $C_j$ contains a positive literal $x_i$ that is evaluated to True, $a_j^C$ receives the edge $(a_i^T, a_j^C)$. 
    Since each of $a_j^C$'s incident edges that she does not receive has a value of 1, $a_j^C$ does not envy other agents after removing the edge with a value of $-1$ from her own bundle or the edge with a value of $1$ from other agents' bundles. 
    
    For the other direction, we assume that the constructed graph has an $\EFX_-^+$ orientation and use the orientation to create a satisfying assignment as follows: for each variable $x_i$, if the edge $(a_i^T, a_i^F)$ is allocated to agent $a_i^T$, then set $x_i$ to True; otherwise, set $x_i$ to False. 
    Next, we show that the assignment is satisfying. 
    First, since the orientation is $\EFX_-^+$, each agent that corresponds to a variable or a clause must receive the edge between herself and $a_1^{\Delta}$ that has a value of $-1$. 
    For each variable $x_i$, if the edge $(a_i^T, a_i^F)$ is allocated to agent $a_i^T$, both $(a_i^F, a_{j^{i, F, 1}}^C)$ and $(a_i^F, a_{j^{i, F, 2}}^C)$ must be allocated to agent $a_i^F$. 
    Otherwise, $a_i^F$ will envy $a_i^T$ even after removing $(a_i^F, a_1^{\Delta})$ from her own bundle. 
    For a similar reason, if the edge $(a_i^T, a_i^F)$ is allocated to $a_i^F$, both $(a_i^T, a_{j^{i, T, 1}}^C)$ and $(a_i^T, a_{j^{i, T, 2}}^C)$ must be allocated to $a_i^T$. 
    For each clause $C_j$, agent $a_j^C$ must receive at least one edge with a value of 1. 
    Otherwise, $a_j^C$ will envy the agents who receive her incident edges that have a value of 1 even after removing the edge with a value of $-1$ from her own bundle. 
    This implies that each clause has a literal that is evaluated to True. 
\end{proof}

Notice that in the graphs constructed in the above reduction, each edge has non-zero values to both its endpoint agents. 
For such graphs, an orientation is $\EFX_0^+$ if and only if it is $\EFX_-^+$, since no agent receives an edge with a value of zero. 
Therefore, the hardness of determining the existence of $\EFX_-^+$ orientations also applies to $\EFX_0^+$ orientations. 

\begin{corollary}
\label{coro:orientation:EFX+0}
    Determining whether an $\EFX^+_0$ orientation exists or not is NP-complete, even for additive valuations. 
\end{corollary}

\subsection{$\EFX$ Orientations on Simple Graphs}
Due to the above hardness results for general graphs, we also study EFX orientations on some simple graphs such as trees, stars and paths. 
For these simple graphs, $\EFX^+_0$ or $\EFX^+_-$ orientations always exist and can be computed in polynomial time. 
Although $\EFX^0_0$ or $\EFX^0_-$ orientations may not exist, determining their existence is in polynomial time. 

\paragraph{Trees}
We first consider trees and have the following result. 
\begin{theorem}\label{thm:orientation:tree}
    For trees, $\EFX^+_0$ or $\EFX^+_-$ orientations always exist and can be computed in polynomial time. 
\end{theorem}
    
Since any $\EFX^+_0$ orientation is also $\EFX^+_-$, it suffices to compute $\EFX^+_0$ orientations. 
To achieve this, we first traverse the tree using the breadth-first search (BFS) algorithm and label each vertex with its layer number during the BFS algorithm. 
We then allocate each edge to the endpoint agent at the upper layer if the edge is a good for that agent, and to the endpoint agent at the lower layer otherwise. 
The formal description is provided in Algorithm \ref{alg:orientation:tree}. 

\begin{algorithm}[tb]
\caption{Computing $\EFX^+_0$ Orientations on Trees}
\label{alg:orientation:tree}
\KwIn{A mixed instance with a tree where $N$ is the vertex set and $M$ is the edge set. }
\KwOut{An $\EFX^+_0$ orientation $\mathbf{X} = (X_1, \ldots, X_n)$.}
Run the breadth-first search algorithm and label each vertex $a_i$ with its layer number $d_i$.\\
$D_l \leftarrow \{a_i \in N: d_i = l\}$ for every $l \in [\max_{a_i\in N}d_i]$. \textcolor{gray}{//$D_l$ contains all vertices at layer $l$}\\
\For{$l$ from 1 to $\max_{a_i\in N}d_i$}{
    \If{$l = 1$}{
        Let $a_k$ be the agent in $D_1$. \\
        $X_k \leftarrow E_k^{>0}$, $M \leftarrow M \setminus E_k^{>0}$. 
    }
    \ElseIf{$l = \max_{a_i\in N}d_i$}{
        \For{each $a_i \in D_l$}{
            Let $e$ be the edge in $M$ incident to $a_i$. \\
            $X_i \leftarrow \{e\}$, $M \leftarrow M \setminus \{e\}$. \\
        }
    }
    \Else{
        \For{each $a_i \in D_l$}{
            Let $e$ be the edge in $M$ that is incident to $a_i$ and an agent in $D_{l-1}$. \\
            $X_i \leftarrow (\{e\} \cup E_i^{>0}) \cap M$. \\
            $M \leftarrow M \setminus (\{e\} \cup E_i^{>0})$. \\
            
        }
    }
}
\Return $\mathbf{X} = (X_1, \ldots, X_n)$. \\
\end{algorithm}

\begin{proof}[Proof of Theorem \ref{thm:orientation:tree}]
    Clearly, Algorithm \ref{alg:orientation:tree} runs in polynomial time. 
    It suffices to show that the computed orientation is $\EFX^+_0$. 
    Observe that for each agent $a_i$, the edge between herself and an agent at one higher layer is the only edge that is a good for her but is allocated to another agent, or the only edge that is not a good for her but is allocated to her. 
    For the first case, after removing the edge from the other agent's bundle, $a_i$ does not envy that agent. 
    For the second case, after removing the edge from her own bundle, $a_i$ does not envy others. 
    Therefore, the orientation is $\EFX^+_0$. 
\end{proof}

Note that stars and paths are also trees, Theorem \ref{thm:orientation:tree} holds for stars and paths. 
However, $\EFX^0_0$ or $\EFX^0_-$ orientations may not exist even for a simple tree that is also a star and a path. 

\begin{proposition}
\label{prop:tree:efx00}
    There exists a simple tree that is also a star and a path, for which no allocation is $\EFX^0_0$ or $\EFX^0_-$. 
\end{proposition}
\begin{proof}
    Consider the simple tree illustrated in Figure \ref{fig:orientation:star_path}, which is also a star and a path. 
    Notice that the edge $e_{2, 3}$ must be allocated to $a_2$. 
    Otherwise, $a_3$ will envy whoever receives $e_{1, 2}$ even after removing $e_{1, 2}$ from the agent's bundle. 
    Moreover, $e_{1, 2}$ must also be allocated to $a_2$. 
    Otherwise, $a_2$ will envy $a_1$ even after removing $e_{2, 3}$ from her own bundle. 
    Since $a_1$ envies $a_2$ even after removing $e_{2, 3}$ from $a_2$'s bundle, no orientation is $\EFX^0_0$ or $\EFX^0_-$. 
\end{proof}

\begin{figure}
    \centering
    \includegraphics[width=0.25\columnwidth]{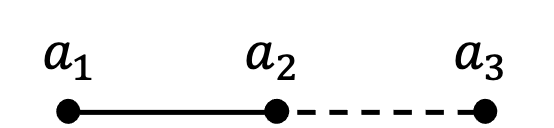}
    \caption{A simple tree (also a star, a path) for which no orientation is $\EFX^0_0$ or $\EFX^0_-$. The dashed edge is a chore for both its endpoint agents. The solid edge is a good for both its endpoint agents.}
    \label{fig:orientation:star_path}
\end{figure}

The good news is that for stars and a specific type of path, we can determine whether $\EFX^0_0$ or $\EFX^0_-$ orientations exist or not in polynomial time. 

\paragraph{Stars}
We then consider stars and have the following result. 
\begin{theorem}\label{thm:orientation:star}
    For stars, determining whether an $\EFX^0_0$ or $\EFX^0_-$ orientation exists or not is in polynomial time. 
\end{theorem}
\begin{proof}
    If there is only one edge, allocating the edge to one of its endpoint agents gives an $\EFX^0_0$ and $\EFX^0_-$ orientation. 
    In the following, we assume that there are at least two edges. 
    Observe that if an edge is a chore for its satellite agent, the edge must be allocated to the center in any $\EFX^0_0$ or $\EFX^0_-$ orientation. 
    To see this, let the edge be $e_{i, j}$ such that $a_j$ is the satellite and $a_i$ is the center and suppose that $e_{i, j}$ is allocated to $a_j$. 
    Let another edge be $e_{i, k}$. 
    Since $e_{i, k}$ is a dummy for $a_j$, no matter who receives $e_{i, k}$, $a_j$ will envy her even after removing $e_{i, k}$ from her bundle.
    Thus, $e_{i, j}$ must be allocated to $a_i$ in any $\EFX^0_0$ and $\EFX^0_-$ orientation.  
    Also, observe that if the center receives an edge, each edge that is a good for its satellite agent must be allocated to the satellite. 
    Otherwise, the satellite will envy the center even after removing the edge that the center has received. 
    We consider the following two cases: 

    \smallskip
    \textbf{\underline{Case 1}:} Each edge is not a chore for its satellite agent. 
    In this case, we simply allocate each edge to its satellite agent, and it is easy to see that the orientation is $\EFX^0_0$ and $\EFX^0_-$. 

    \smallskip
    \textbf{\underline{Case 2}:} There exists an edge that is a chore for its satellite agent. 
    In this case, as we have seen, the edges that are chores for their satellite agents must be allocated to the center in any $\EFX^0_0$ and $\EFX^0_-$ orientation.
    Moreover, the edges that are goods for their satellite agents must be allocated to their satellites. 
    For the edges that are dummies for their satellite agents, it does not matter how they are allocated since the orientation is always envy-free for their satellites. 
    Thus, if they are not goods for the center, we allocate them to their satellites and otherwise we allocate them to the center. 
    Therefore, by checking whether the orientation is $\EFX^0_0$ (resp. $\EFX^0_-$) for the center, we can determine whether there exists any $\EFX^0_0$ (resp. $\EFX^0_-$) orientation. 
\end{proof}

\paragraph{Paths} We next consider a specific type of paths and have the following result. 
\begin{theorem}\label{thm:orientation:special_path}
    For paths where each edge is either a good or a chore for both its endpoint agents, determining whether an $\EFX^0_0$ or $\EFX^0_-$ orientation exists or not is in polynomial time. 
\end{theorem}
\begin{proof}
    In the paths we focus on, $\EFX^0_0$ and $\EFX^0_-$ are equivalent since no edge is a dummy for any of its endpoint agents. 
    Thus, in the following, we only consider $\EFX^0_-$. 

    If there is only one edge or each edge is a good for both its endpoint agents, we can obtain an $\EFX^0_-$ orientation by allocating each edge to the right endpoint agent. 
    If there are only two edges and at least one is a chore for both its endpoint agents, an $\EFX^0_-$ orientation does not exist, which is easy to see by the same reasoning we made in the proof of Proposition \ref{prop:tree:efx00}. 
    It remains to consider the cases where there are at least three edges and at least one is a chore for both its endpoint agents. 
    
    We first observe that for each edge that is a chore for both its endpoint agents, its neighboring vertices and edges follow the same pattern if there exists an $\EFX^0_-$ orientation. 
    Let one such edge be $e_{i_0, i_1}$ that is a chore for both $a_{i_0}$ and $a_{i_1}$, the two neighboring edges on its right are $e_{i_1, i_2}$, $e_{i_2, i_3}$, and the two neighboring edges on its left are $e_{i_{-2}, i_{-1}}$, $e_{i_{-1}, i_0}$. 
    We claim that if there exists an $\EFX^0_-$ orientation, one of the following two conditions holds:
    \begin{enumerate}
        \item $e_{i_1, i_2}$ is a good for both $a_{i_1}$ and $a_{i_2}$, $e_{i_2, i_3}$ is a good for both $a_{i_2}$ and $a_{i_3}$, $v_{i_1}(\{e_{i_0, i_1}\} \cup \{e_{i_1, i_2}\}) \ge 0$, $v_{i_2}(e_{i_1, i_2}) \le v_{i_2}(e_{i_2, i_3})$; 
        \item $e_{i_{-1}, i_0}$ is a good for both $a_{i_{-1}}$ and $a_{i_0}$, $e_{i_{-2}, i_{-1}}$ is a good for both $a_{i_{-2}}$ and $a_{i_{-1}}$, $v_{i_0}(\{e_{i_0, i_1}\} \cup \{e_{i_{-1}, i_0}\}) \ge 0$, $v_{i_{-1}}(e_{i_{-1}, i_0}) \le v_{i_{-1}}(e_{i_{-2}, i_{-1}})$. 
    \end{enumerate}
    To see this, first consider the case if $e_{i_0, i_1}$ is allocated to $a_{i_1}$. 
    If $e_{i_1, i_2}$ is not a good for $a_{i_1}$ or $v_{i_1}(\{e_{i_0, i_1}\} \cup \{e_{i_1, i_2}\}) < 0$, $a_{i_1}$ gets a negative value and will envy whoever gets an edge that is not incident to her even after removing that edge. 
    Therefore, in any $\EFX^0_-$ orientation, $e_{i_1, i_2}$ must be a good for $a_{i_1}$, $v_{i_1}(\{e_{i_0, i_1}, e_{i_1, i_2}\}) \ge 0$ and $e_{i_1, i_2}$ must be allocated to $a_{i_1}$. 
    Furthermore, if $e_{i_2, i_3}$ is not a good for $a_{i_2}$ or $v_{i_2}(e_{i_1, i_2}) > v_{i_2}(e_{i_2, i_3})$, $a_{i_2}$ will envy $a_{i_1}$ after removing $e_{i_0, i_1}$ from her bundle. 
    Therefore, in any $\EFX^0_-$ orientation, $e_{i_2, i_3}$ must be a good for $a_{i_2}$, $v_{i_2}(e_{i_1, i_2}) \le v_{i_2}(e_{i_2, i_3})$ and $e_{i_2, i_3}$ must be allocated to $a_{i_2}$. 
    An analogous statement holds if $e_{i_0, i_1}$ is allocated to $a_{i_0}$. 

    We are now ready to determine whether an $\EFX^0_-$ orientation exists or not. 
    We consider three cases: 

    \smallskip
    \textbf{\underline{Case 1}:}
    If there exists an edge such that it is a chore for both its endpoint agents but its neighboring vertices and edges on neither of its two sides follow the pattern we observed (i.e., satisfy the above two conditions), we can assert that there does not exist any $\EFX^0_-$ orientation. 

    \smallskip
    \textbf{\underline{Case 2}:}
    If for each edge that is a chore for both its endpoint agents, the neighboring vertices and edges on both its two sides follow the pattern, we allocate each edge that is a chore to the left endpoint agent and each other edge to the right one. 
    It is easy to verify that the orientation is $\EFX^0_-$. 

    \smallskip
    \textbf{\underline{Case 3}:}
    If there exists an edge that is a chore for both its endpoints, and its neighboring vertices and edges on only one of its two sides follow the pattern, we know how these edges must be allocated in any $\EFX^0_-$ orientation. 
    Take the edge $e_{i_0, i_1}$ for example, if only the first condition holds, $e_{i_0, i_1}$ and $e_{i_1, i_2}$ must be allocated to $a_{i_1}$, and $e_{i_2, i_3}$ must be allocated to $a_{i_2}$. 
    Therefore, we can delete these three edges and $a_{i_1}$ and $a_{i_2}$ from the path. 
    It is easy to see that the original path has an $\EFX^0_-$ orientation if and only if all reduced paths have $\EFX^0_-$ orientations. 
    By repeating the above reduction, either in some sub-path, Case 1 occurs and we can assert that there does not exist any $\EFX^0_-$ orientation; or in all sub-paths, either Case 2 occurs or each edge is a good for both its endpoint, then we can conclude that an $\EFX^0_-$ orientation exists. 
\end{proof}

\section{EFX Allocations for Mixed Instances}
\label{sec:mixed:alloc}
In this section, we elaborate on EFX allocations for mixed instances. 

\subsection{$\EFX_0^0$ Allocations}
\label{subsec:allocation:EFX00}
We start with the strongest one among those four notions, i.e., $\EFX^0_0$. 
We say an edge $e$ is \textit{priceless} to an agent $a_i$ if for any $S_1, S_2 \subseteq M$ such that $e\notin S_1$ and $e\in S_2$, we have $v_i(S_1) < v_i(S_2)$.
Notice that an agent envies whoever receives her priceless edge no matter which edges she herself receives and doe not envy others as long as she receives it. 
We have the following proposition, which provides some characterization of $\EFX_0^0$ allocations on some graphs with priceless edges.

\begin{proposition}\label{pro:allocation:EFX00:priceless}
    For graphs that satisfy (1) each edge is a good for both its endpoint agents, (2) each agent has one priceless incident edge and (3) each priceless edge is priceless to both its endpoint agents, we have that each edge must be allocated to one of its endpoint agents in any $\EFX_0^0$ allocation. 
\end{proposition}
\begin{proof}
    First, consider priceless edges. 
    Let $(a_i, a_j)$ be one priceless edge and assume for the sake of contradiction that it is allocated to agent $a_k$ who is not $a_i$ or $a_j$. 
    In this case, $a_k$ is envied by $a_i$ and $a_j$, and thus cannot receive her priceless edge since it is a dummy for $a_i$ and $a_j$. 
    Otherwise, $a_i$ and $a_j$ will envy $a_k$ even after removing the dummy from $a_k$'s bundle. 
    Moreover, $a_k$ envies whoever receives her priceless edge even after removing $(a_i, a_j)$ from her own bundle, which is a dummy for her, thus the allocation is not $\EFX^0_0$. 

    Then, assume that each priceless edge is allocated to one of its endpoint agents and consider non-priceless edges. 
    Observe that each agent either envies the agent who receives the edge that is priceless to her, or is envied by that agent if she receives that priceless edge. 
    Each envied agent cannot receive any other edge since every edge is either a good or a dummy for the agent who envies her. 
    Otherwise, after removing the other edge from the envied agent's bundle, the envious agent still envies the envied agent, thus the allocation is not $\EFX^0_0$. 
    Each agent who envies others cannot receive any edge that is not incident to her since it is a dummy for her. 
    Otherwise, after removing the dummy from her own bundle, the envious agent still envies others, thus the allocation is not $\EFX^0_0$. 
    Therefore, each non-priceless edge must also be allocated to one of its endpoint agents. 
\end{proof}

We first show that $\EFX^0_0$ allocations may not exist. 

\begin{proposition}
\label{pro:allocation:EFX00_not_exist}
    There exist graphs for which no allocation is $\EFX^0_0$. 
\end{proposition}
\begin{proof}
    Consider the graph illustrated in Figure \ref{fig:allocation:EFX00:not_exist}. 
    Notice that the graph satisfies the three properties in Proposition \ref{pro:allocation:EFX00:priceless}. 
    Therefore, each edge must be allocated to one of its endpoint agents in any $\EFX^0_0$ allocation. 
    Without loss of generality, assume that $e_{1, 2}$ is allocated to $a_1$, $e_{3, 4}$ to $a_3$, and $e_{1, 3}$ to $a_1$. 
    Since $a_2$ envies $a_1$ even after removing $e_{1,3}$ from $a_1$'s bundle, the allocation is not $\EFX^0_0$. 
\end{proof}

\begin{figure}[tb]
    \centering
    \includegraphics[width=0.2\columnwidth]{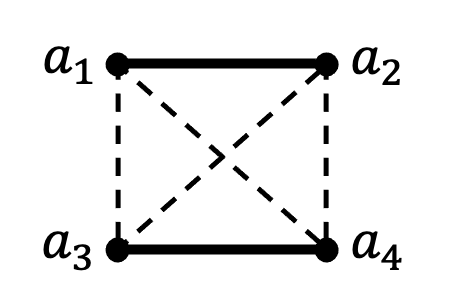}
    \caption{An example for which no allocation is $\EFX^0_0$. Each bold solid edge is priceless to both its endpoint agents. }
    \label{fig:allocation:EFX00:not_exist}
\end{figure}

We next study the complexity of determining the existence of $\EFX^0_0$ allocations and have the following result. 

\begin{theorem}\label{thm:allocation:EFX00}
    Determining whether an $\EFX_0^0$ allocation exists or not is NP-complete, even for additive valuations. 
\end{theorem}

To prove Theorem \ref{thm:allocation:EFX00}, we reduce from Circuit-SAT problem to the $\EFX_0^0$ allocation problem. 
Circuit-SAT problem determines whether a given Boolean circuit has an assignment of the inputs that makes the output True, which is well-known to be NP-complete \citep{karp1975reducibility}. 

\begin{figure}[tb]
\begin{subfigure}{\columnwidth}
    \centering
    \includegraphics[width=0.4\linewidth]{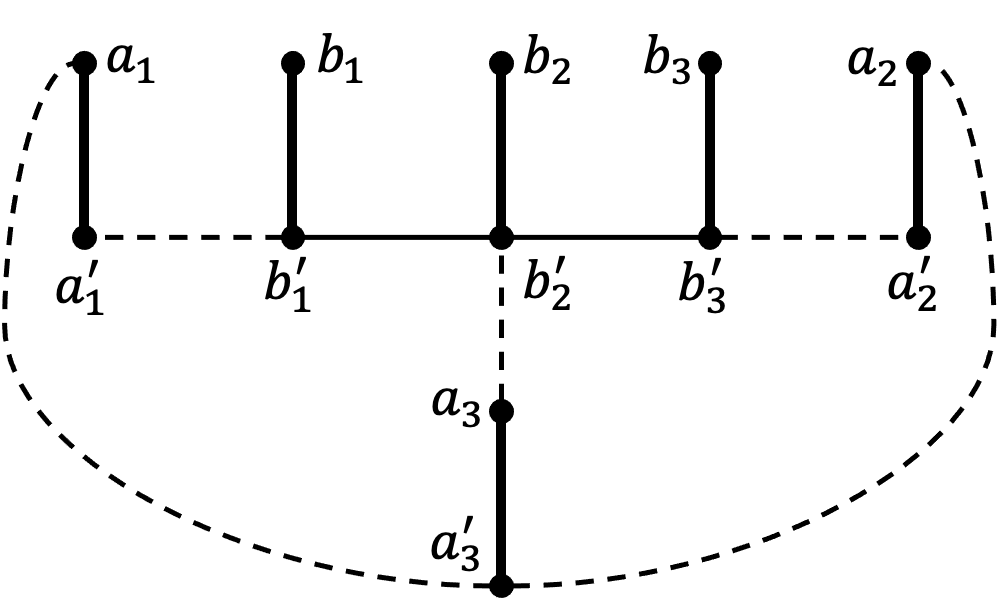}
    \caption{}
    \label{fig:allocation:EFX00:gadgets:OR}
\end{subfigure}

\hspace{0.32cm}
\begin{subfigure}{0.3\columnwidth}
    \centering
    \includegraphics[width=0.68\linewidth]{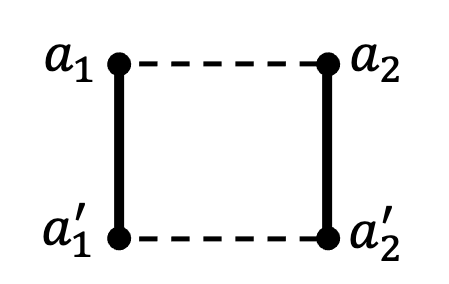}
    \caption{}
    \label{fig:allocation:EFX00:gadgets:NOT}
\end{subfigure}
\begin{subfigure}{0.3\columnwidth}
    \centering
    \includegraphics[width=0.68\linewidth]{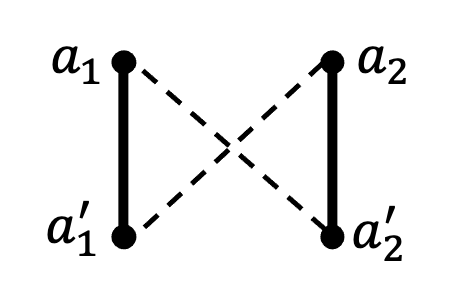}
    \caption{}
    \label{fig:allocation:EFX00:gadgets:WIRE}
\end{subfigure}
\begin{subfigure}{0.3\columnwidth}
    \centering
    \includegraphics[width=0.68\linewidth]{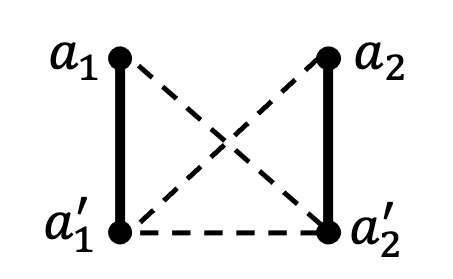}
    \caption{}
    \label{fig:allocation:EFX00:gadgets:TRUE}
\end{subfigure}
\caption{(a) OR gadget, (b) NOT gadget, (c) WIRE gadget, (d) TRUE terminator gadget. In these graphs, each agent has an additive valuation. Each bold solid edge is priceless to both its endpoint agents (e.g., it has an infinitely large value of $+\infty$), each non-bold solid edge has an infinitely small value of $\epsilon_1 > 0$ to both its endpoint agents, each dashed line also has an infinitely small value of $\epsilon_2$ to both its endpoint agents with $\epsilon_1 >\epsilon_2 > 0$.}
\label{fig:allocation:EFX00:gadgets}
\end{figure}

We first show how to simulate the OR gate, the NOT gate, the wire in the circuit and how to force the final output to be True. 
To achieve this, we construct four graphs, named OR gadget, NOT gadget, WIRE gadget, TRUE terminator gadget, respectively (see Figure \ref{fig:allocation:EFX00:gadgets}). 
It is easy to see that Proposition \ref{pro:allocation:EFX00:priceless} applies to all these four gadgets. 
That is, in any $\EFX_0^0$ allocation on each of these gadgets, each edge must be allocated to one of its endpoint agents. 
This enables us to represent each input (or output) in the circuit as an edge in the gadgets and its value (True or False) as the orientation of the edge. 

In the OR gadget (see Figure \ref{fig:allocation:EFX00:gadgets:OR}), edges $(a_1, a_1^{\prime})$ and $(a_2, a_2^{\prime})$ represent the two inputs of the OR gate, edge $(a_3, a_3^{\prime})$ represents the output. 
The following claim shows that the OR gadget correctly simulates the OR gate. 

\begin{claim}\label{clm:allocation:EFX00:OR}
    In every $\EFX_0^0$ allocation on the OR gadget, edge $(a_3, a_3^{\prime})$ is allocated to $a_3$ if and only if edge $(a_1, a_1^{\prime})$ is allocated to $a_1$ or edge $(a_2, a_2^{\prime})$ is allocated to $a_2$.
\end{claim}
\begin{proof}
    We first show that if $(a_1, a_1^{\prime})$ is allocated to $a_1$, $(a_3, a_3^{\prime})$ must be allocated to $a_3$. 
    Since $(a_1, a_1^{\prime})$ is priceless to $a_1^{\prime}$ but is allocated to $a_1$, $a_1^{\prime}$ envies $a_1$. 
    Hence, $(a_1, a_3^{\prime})$ must be allocated to $a_3^{\prime}$. 
    Otherwise, $a_1^{\prime}$ still envies $a_1$ after removing $(a_1, a_3^{\prime})$ from $a_1$'s bundle. 
    Moreover, $(a_3, a_3^{\prime})$ must be allocated to $a_3$.
    Otherwise, $a_3$ still envies $a_3^{\prime}$ after removing $(a_1, a_3^{\prime})$ from $a_3$'s bundle. 
    By symmetry, it holds that if $(a_2, a_2^{\prime})$ is allocated to $a_2$, $(a_3, a_3^{\prime})$ must be allocated to $a_3$.

    We then show that when $(a_1, a_1^{\prime})$ is allocated to $a_1$ and $(a_3, a_3^{\prime})$ is allocated to $a_3$, no matter which endpoint agent $(a_2, a_2^{\prime})$ is allocated to, there exists an $\EFX_0^0$ allocation. 
    When $(a_2, a_2^{\prime})$ is allocated to $a_2$, we construct an $\EFX_0^0$ allocation as follows: allocate each priceless edge to the upper endpoint agent, i.e., $(a_i, a_i^{\prime})$ to $a_i$ for every $i \in \{1, 2, 3\}$ and $(b_i, b_i^{\prime})$ to $b_i$ for every $i \in \{1, 2, 3\}$; 
    allocate the middle four edges to the endpoint agents who are further away from $b_2^{\prime}$, i.e., $(a_1^{\prime}, b_1^{\prime})$ to $a_1^{\prime}$, $(b_1^{\prime}, b_2^{\prime})$ to $b_1^{\prime}$, $(b_2^{\prime}, b_3^{\prime})$ to $b_3^{\prime}$, $(b_3^{\prime}, a_2^{\prime})$ to $a_2^{\prime}$; 
    allocate $(b_2^{\prime}, a_3)$ to $b_2^{\prime}$, $(a_1, a_3^{\prime})$ to $a_3^{\prime}$, $(a_2, a_3^{\prime})$ to $a_3^{\prime}$. 
    Since each agent has a positive value for each edge she receives, to verify that the allocation is $\EFX_0^0$, it suffices to consider the agents who receive more than one edge (only $a_3^{\prime}$ in the above allocation). 
    Since both $a_1$ and $a_2$ receive their priceless edges, neither of them envies $a_3^{\prime}$ and thus the allocation is $\EFX_0^0$.
    When $(a_2, a_2^{\prime})$ is allocated to $a_2^\prime$, we construct an $\EFX_0^0$ allocation as follows: 
    allocate each priceless edge except $(a_2, a_2^{\prime})$ and $(b_1, b_1^{\prime})$  to the upper endpoint, i.e., $(a_i, a_i^{\prime})$ to $a_i$ for every $i \in \{1, 3\}$, $(b_i, b_i^{\prime})$ to $b_i$ for every $i \in \{2, 3\}$, $(a_2, a_2^{\prime})$ to $a_2^{\prime}$, $(b_1, b_1^{\prime})$ to $b_1^{\prime}$;
    for the middle four edges, allocate $(a_1^{\prime}, b_1^{\prime})$ to $a_1^{\prime}$, $(b_1^{\prime}, b_2^{\prime})$ to $b_2^{\prime}$, $(b_2^{\prime}, b_3^{\prime})$ to $b_3^{\prime}$, $(b_3^{\prime}, a_2^{\prime})$ to $b_3^{\prime}$; 
    allocate $(b_2^{\prime}, a_3)$ to $b_2^{\prime}$, $(a_1, a_3^{\prime})$ to $a_3^{\prime}$, $(a_2, a_3^{\prime})$ to $a_2$. 
    In the above allocation, only $b_2^{\prime}$ and $b_3^{\prime}$ receive more than one edge. 
    For $b_2^{\prime}$, neither $b_1^{\prime}$ nor $a_3$ envies her since both of them receive their priceless edges. 
    For $b_3^{\prime}$, $a_2^{\prime}$ does not envy her since she receives her priceless edge, and $b_2^{\prime}$ does not envy her since she receives a value of $\epsilon_1 + \epsilon_2$ and thinks that $b_3^{\prime}$ receives a value of $\epsilon_1$. 
    Therefore, the allocation is also $\EFX_0^0$. 
    By symmetry, when $(a_2, a_2^{\prime})$ is allocated to $a_2$ and $(a_3, a_3^{\prime})$ is allocated to $a_3$, no matter which endpoint agent $(a_1, a_1^{\prime})$ is allocated to, there exists an $\EFX_0^0$ allocation.

    We next show that if both $(a_1, a_1^{\prime})$ and $(a_2, a_2^{\prime})$ are allocated to their lower endpoint agents, $(a_3, a_3^{\prime})$ must be allocated to $a_3^{\prime}$. 
    It suffices to show that $(b_2^{\prime}, a_3)$ must be allocated to $a_3$. 
    This is because if both $(a_3, a_3^{\prime})$ and $(b_2^{\prime}, a_3)$ are allocated to $a_3$, $a_3^{\prime}$ will envy $a_3$ even after removing $(b_2^{\prime}, a_3)$ from $a_3$'s bundle. 
    If $(b_2, b_2^{\prime})$ is allocated to $b_2^{\prime}$, $(b_2^{\prime}, a_3)$ must be allocated to $a_3$ and we have done, since otherwise $b_2$ will envy $b_2^{\prime}$ even after removing $(b_2^{\prime}, a_3)$ from $b_2^{\prime}$'s bundle. 
    Therefore, it remains to consider the case when $(b_2, b_2^{\prime})$ is allocated to $b_2$. 
    Since $(a_1, a_1^{\prime})$ is allocated to $a_1^{\prime}$, $(a_1^{\prime}, b_1^{\prime})$ must be allocated to $b_1^{\prime}$ since otherwise $a_1$ will envy $a_1^{\prime}$ even after removing $(a_1^{\prime}, b_1^{\prime})$ from $a_1^{\prime}$'s bundle. 
    Furthermore, $(b_1, b_1^{\prime})$ must be allocated to $b_1$. 
    By the same reasoning, $(a_2^{\prime}, b_3^{\prime})$ must be allocated to $b_3^{\prime}$ and $(b_3, b_3^{\prime})$ must be allocated to $b_3$. 
    Then consider the incident edges of $b_2^{\prime}$ that have not been allocated so far, i.e., $(b_1^{\prime}, b_2^{\prime})$ and $(b_2^{\prime}, b_3^{\prime})$.  $b_2^{\prime}$ must receive one of these two edges, since otherwise she will envy $b_1^{\prime}$ even after removing $(a_1^{\prime}, b_1^{\prime})$ from $b_1^{\prime}$'s bundle, and $b_3^{\prime}$ even after removing $(a_2^{\prime}, b_3^{\prime})$ from $b_3^{\prime}$'s bundle. 
    No matter which edge $b_2^{\prime}$ receives, $(b_2^{\prime}, a_3)$ must be allocated to $a_3$. 
    To see this, let the edge that $b_2^{\prime}$ receives be $(b_1^{\prime}, b_2^{\prime})$. 
    Since $b_1^{\prime}$ receives a value of $\epsilon_2$ and thinks that $b_2^{\prime}$ receives a value of $\epsilon_1 > \epsilon_2$, she envies $b_2^{\prime}$ and thus $b_2^{\prime}$ cannot receive $(b_2^{\prime}, a_3)$ any more. 

    Lastly, we show that when all of $(a_1, a_1^{\prime})$, $(a_2, a_2^{\prime})$ and $(a_3, a_3^{\prime})$ are allocated to their lower endpoint agents, there exists an $\EFX_0^0$ allocation. 
    We allocate each priceless edge except $(b_1, b_1^{\prime})$ and $(b_3, b_3^{\prime})$ to the lower endpoint agent, i.e., $(a_i, a_i^{\prime})$ to $a_i^{\prime}$ for every $i \in \{1, 2, 3\}$, $(b_2, b_2^{\prime})$ to $b_2^{\prime}$, and $(b_i, b_i^{\prime})$ to $b_i$ for every $i \in \{1, 3\}$; 
    for the middle four edges, allocate $(a_1^{\prime}, b_1^{\prime})$ and $(b_1^{\prime}, b_2^{\prime})$ to $b_1^{\prime}$, $(b_2^{\prime}, b_3^{\prime})$ and $(b_3^{\prime}, a_2^{\prime})$ to $b_3^{\prime}$; 
    allocate $(b_2^{\prime}, a_3)$ to $a_3$, $(a_1, a_3^{\prime})$ to $a_1$, $(a_2, a_3^{\prime})$ to $a_2$. 
    In the above allocation, only $b_1^{\prime}$ and $b_3^{\prime}$ receive more than one edge. 
    For $b_1^{\prime}$, neither $a_1^{\prime}$ nor $b_2^{\prime}$ envies her since both of them receive their priceless edges. 
    By the same reasoning, neither $b_2^{\prime}$ nor $a_2^{\prime}$ envies $b_3^{\prime}$. 
    Therefore, the allocation is $\EFX_0^0$. 
\end{proof}

In the NOT gadget (see Figure \ref{fig:allocation:EFX00:gadgets:NOT}), edge $(a_1, a_1^{\prime})$ represents the input of the NOT gate, and edge $(a_2, a_2^{\prime})$ represents the output. 
The following claim shows that the NOT gadget correctly simulates the NOT gate. 

\begin{claim}\label{clm:allocation:EFX00:NOT}
    In every $\EFX_0^0$ allocation on the NOT gadget, edge $(a_2, a_2^{\prime})$ is allocated to $a_2^{\prime}$ if and only if edge $(a_1, a_1^{\prime})$ is allocated to $a_1$.
\end{claim}
\begin{proof}
    If $(a_1, a_1^{\prime})$ is allocated to $a_1$, $(a_1, a_2)$ must be allocated to $a_2$ since otherwise $a_1^{\prime}$ will envy $a_1$ even after removing $(a_1, a_2)$ from $a_1$'s bundle. 
    Furthermore, $(a_2, a_2^{\prime})$ must be allocated to $a_2^{\prime}$. 
    By symmetry, if $(a_1, a_1^{\prime})$ is allocated to $a_1^{\prime}$,  $(a_2, a_2^{\prime})$ must be allocated to $a_2$. 
    When $(a_1, a_1^{\prime})$ is allocated to $a_1$ and $(a_2, a_2^{\prime})$ is allocated to $a_2^{\prime}$, allocating the edges clockwise produces an $\EFX_0^0$ allocation. 
    When $(a_1, a_1^{\prime})$ is allocated to $a_1^{\prime}$ and $(a_2, a_2^{\prime})$ is allocated to $a_2$, allocating the edges counterclockwise produces an $\EFX_0^0$ allocation. 
\end{proof}

In the WIRE gadget (see Figure \ref{fig:allocation:EFX00:gadgets:WIRE}), edge $(a_1, a_1^{\prime})$ represents the input of the wire of a circuit, and edge $(a_2, a_2^{\prime})$ represents the output. 
Since the only difference between the NOT gadget and the WIRE gadget is that the labels of $a_2$ and $a_2^{\prime}$ are exchanged, we have the following claim that shows the WIRE gadget correctly simulates the wire in the circuit. 

\begin{claim}\label{clm:allocation:EFX00:WIRE}
    In every $\EFX_0^0$ allocation on the WIRE gadget, edge $(a_2, a_2^{\prime})$ is allocated to $a_2$ if and only if edge $(a_1, a_1^{\prime})$ is allocated to $a_1$.
\end{claim}

In the TRUE terminator gadget (see Figure \ref{fig:allocation:EFX00:gadgets:TRUE}), edge $(a_2, a_2^{\prime})$ is allocated to $a_2$ in every $\EFX_0^0$ allocation. 

\begin{claim}\label{clm:allocation:EFX00:TRUE}
    In every $\EFX_0^0$ allocation on the TRUE terminator gadget, edge $(a_2, a_2^{\prime})$ is allocated to $a_2$.
\end{claim}
\begin{proof}
    For the sake of contradiction, suppose that $(a_2, a_2^{\prime})$ is allocated to $a_2^{\prime}$. 
    In this case, $(a_1, a_2^{\prime})$ must be allocated to $a_1$, and $(a_1^{\prime}, a_2^{\prime})$ must be allocated to $a_1^{\prime}$. 
    Otherwise, $a_2$ will envy $a_2^{\prime}$ even after removing one edge except $(a_2, a_2^{\prime})$ from $a_2^{\prime}$'s bundle. 
    Then neither $a_1$ nor $a_1^{\prime}$ can receive  $(a_1, a_1^{\prime})$, a contradiction. 
    This is because the endpoint agent that receives $(a_1, a_1^{\prime})$ will be envied by the other even after removing one edge except $(a_1, a_1^{\prime})$ from her bundle. 
    When $(a_2, a_2^{\prime})$ is allocated to $a_2$, allocating $(a_1, a_1^{\prime})$ to $a_1$, $(a_1^{\prime}, a_2)$ and $(a_1^{\prime}, a_2^{\prime})$ to $a_1^{\prime}$, $(a_1, a_2^{\prime})$ to $a_2^{\prime}$ produces an $\EFX_0^0$ allocation. 
\end{proof}

Given a circuit, we first substitute each AND gate with three NOT gates and one OR gate, and get an equivalent circuit without AND gates. 
For the new circuit, we construct a priceless edge with a value of $+\infty$ for each input, and the corresponding gadget for each gate and wire. 
We then construct a True terminator gadget to force the final output to be True. 
Figure \ref{fig:allocation:EFX00:AND_circuit} shows the graph constructed from a simple circuit with one AND gate, two inputs and one final output.
Note that Proposition \ref{pro:allocation:EFX00:priceless} still applies to the graph we construct. 

Up to now, it is not hard to see the correctness of Theorem \ref{thm:allocation:EFX00}. 

\begin{proof}[Proof of Theorem \ref{thm:allocation:EFX00}]
    To prove Theorem \ref{thm:allocation:EFX00}, we show that the circuit has a satisfying assignment if and only if the constructed graph has an $\EFX^0_0$ allocation. 
    For one direction, we assume that there exists an assignment of the inputs such that the final output of the circuit is True and use the assignment to create an allocation as follows: for each input, if it is set to True in the assignment, allocate the corresponding edge to the upper endpoint; otherwise, allocate the edge to the lower endpoint. 
    Allocate the edge that simulates the final output to the upper endpoint. 
    Allocate the remaining edges in each gadget according to Claims \ref{clm:allocation:EFX00:OR} to \ref{clm:allocation:EFX00:TRUE}. 
    Clearly, the allocation is $\EFX_0^0$. 
    
    For the other direction, we assume that there exists an $\EFX_0^0$ allocation in the constructed graph and use the allocation to create an assignment as follows: for each input, if the corresponding edge is allocated to the upper endpoint, set the input to True; otherwise, set the input to False. 
    By Claim \ref{clm:allocation:EFX00:TRUE}, the edge that simulates the final output must be allocated to the upper endpoint. 
    Therefore, by Claims \ref{clm:allocation:EFX00:OR} to \ref{clm:allocation:EFX00:WIRE}, the assignment makes the final output of the circuit True. 
\end{proof}

\begin{figure}[tb]
\centering
\includegraphics[width=0.6\columnwidth]{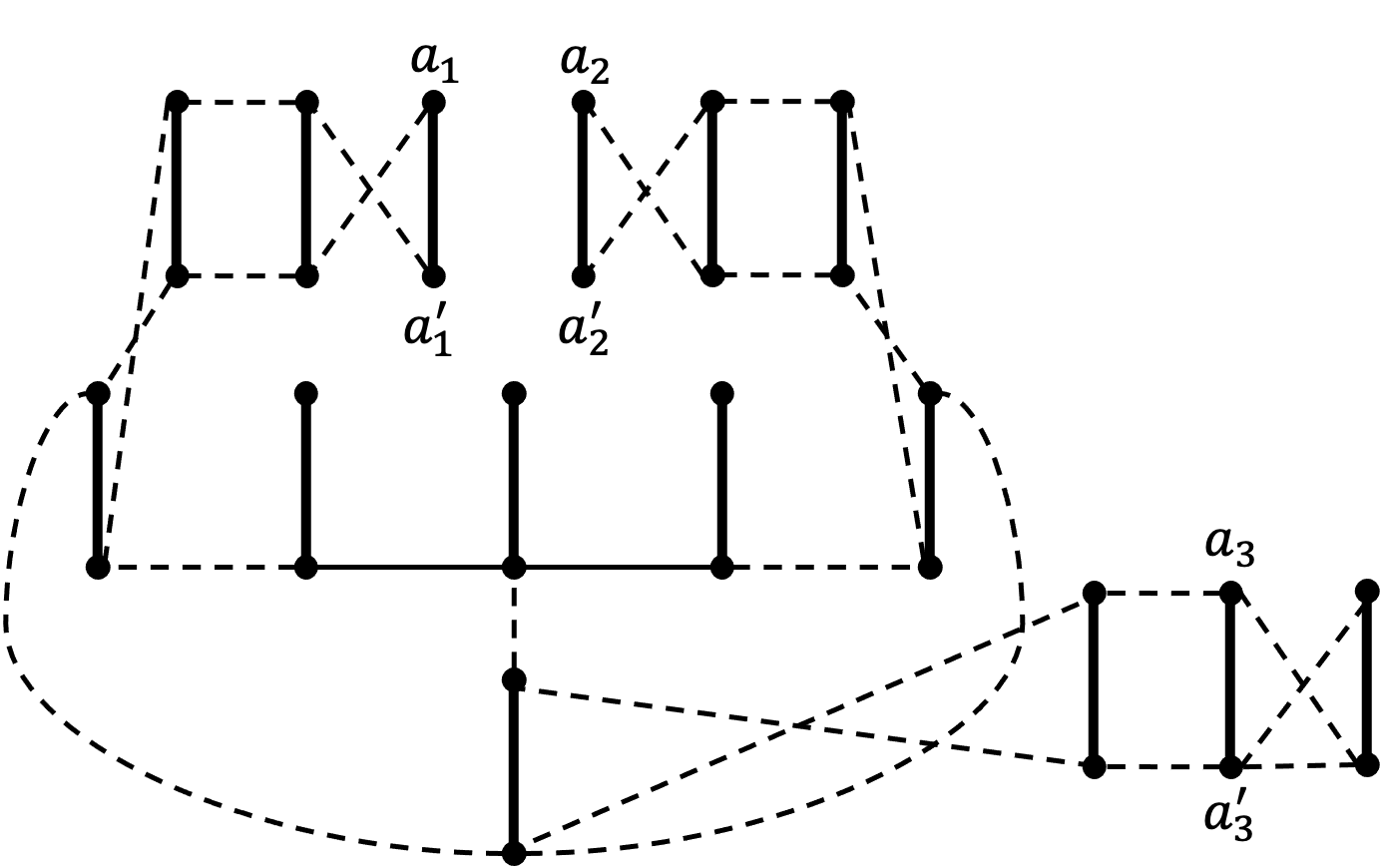}
\caption{The graph is constructed from the circuit that consists of only one AND gate, two inputs, and one final output. $(a_1, a_1^{\prime})$ and $(a_2, a_2^{\prime})$ simulate the inputs, $(a_3, a_3^{\prime})$ simulates the final output. }
\label{fig:allocation:EFX00:AND_circuit}
\end{figure}

\begin{remark}
    Our reduction borrows an idea from the reduction by \citet{christodoulou2023fair} (see Theorem 2 in their paper) and generalizes their reduction. 
    Our reduction can imply their result, while theirs cannot carry over to our problem since it relies on the orientation model. 
\end{remark}

\subsection{$\EFX_-^0$ Allocations}
\label{subsec:allocation:EFX0-}
We next study $\EFX_-^0$ and have the following theorem. 
\begin{theorem}\label{thm:allocation:EFX0-}
    For any simple graph, an $\EFX_-^0$ allocation always exists and can be computed in polynomial time. 
\end{theorem}
 
We first introduce some notations. 
Given a (partial) allocation $\mathbf{X}  = (X_1, \ldots, X_n)$, let $R(\mathbf{X})$ denote the set of unallocated edges, i.e., $R(\mathbf{X}) = M \setminus \bigcup_{a_i \in N}X_i$.
We say an agent $a_j$ is \textit{safe} for another agent $a_i$ if $a_i$ does not envy $a_j$ even if $a_j$ receives all $a_i$'s unallocated incident edges that are not chores for $a_i$, i.e., $v_i(X_i) \ge v_i(X_j \cup (E_i^{\ge 0} \cap R(\mathbf{X})))$. 
We next introduce some properties of allocations. 

\begin{definition}[Properties of a (Partial) Allocation]\label{def:allocation:EFX0-:properties}
    We say that a (partial) allocation $\mathbf{X}$ satisfies
    \begin{itemize}
        \vspace{-4pt}
        \item Property (1) if for every agent $a_i$, the value of her bundle is at least the largest value among her unallocated incident edges that are not chores for her. That is, $v_i(X_i) \ge v_i(e)$ for every edge $e \in E_i^{\ge 0} \cap R(\mathbf{X})$;
        \vspace{-4pt}
        \item Property (2) if for every envied agent $a_i$, the value of her bundle is at least the value of all her unallocated incident edges that are not chores for her. That is, $v_i(X_i) \ge v_i(E_i^{\ge 0} \cap R(\mathbf{X}))$; 
        \vspace{-4pt}
        \item Property (3) if for every two envied agents, there exists a non-envied agent who is safe for both of them; 
        \vspace{-4pt}
        \item Property (4) if no agent receives an edge that is a chore for her. That is, $e \in E_i^{\ge 0}$ for any $a_i \in N$ and $e \in X_i$; 
        \vspace{-4pt}
        \item Property (5) if every envied agent $a_i$ receives exactly one edge, i.e., $|X_i| = 1$; 
        \vspace{-4pt}
        \item Property (6) if every envied agent is envied by exactly one agent; 
        \vspace{-4pt}
        \item Property (7) if there is no envy cycle among the agents. That is, there does not exist a sequence of the agents $a_{i_0} \leftarrow a_{i_1} \leftarrow \cdots \leftarrow a_{i_s}$ such that $a_{i_l}$ envies $a_{i_{l-1}}$ for every $l \in [s]$ and $i_0 = i_s$; 
        \vspace{-4pt}
        \item Property (8) if for any sequence of agents $a_{i_0} \leftarrow a_{i_1} \leftarrow \cdots \leftarrow a_{i_s}$ such that  $a_{i_l}$ envies $a_{i_{l-1}}$ for every $l \in [s]$ and $a_{i_s}$ is non-envied, we have that $a_{i_l}$ is safe for $a_{i_0}$ for every $l \in [s]$.
        \vspace{-4pt}
    \end{itemize}
\end{definition}

We obtain an $\EFX^0_-$ allocation in two parts. 

\noindent \textbf{Part 1}. 
In the first part, we compute a (partial) $\EFX^0_-$ orientation that satisfies Properties (1)-(8) in Definition \ref{def:allocation:EFX0-:properties}. 
Our algorithms in this part are adapted from those by \cite{christodoulou2023fair}. 
There are two differences between our algorithms and \cite{christodoulou2023fair}'s. 
First, since there is one more requirement in our problem that agents cannot envy others after removing a chore from their own bundles, we need to carefully allocate the edges that are chores for their endpoint agents. 
Second, the algorithms by \cite{christodoulou2023fair} cannot guarantee Property (8) and our algorithms need to deal with the case where Property (8) is not satisfied.

We have the following lemma. 
The algorithms and proofs can be seen in Appendix \ref{ap:allocation:EFX0-:part1}. 

\begin{lemma}\label{lem:allocation:EFX0-:first}
    For any simple graph, a (partial) $\EFX^0_-$ orientation that satisfies Properties (1)-(8) in Definition \ref{def:allocation:EFX0-:properties} can be computed in polynomial time. 
\end{lemma}

\noindent \textbf{Part 2}. 
In the second part, we allocate the edges that are not allocated in Part 1. 
We first categorize the unallocated edges into four disjoint groups: 
\begin{itemize}
    \vspace{-4pt}
    \item $G_1$ contains each edge that has at least one non-envied endpoint agent for whom the edge is not a chore;
    \vspace{-4pt}
    \item $G_2$ contains each edge that has two envied endpoints; 
    \vspace{-4pt}
    \item $G_3$ contains each edge that has one non-envied endpoint agent for whom the edge is a chore and one envied endpoint agent for whom it is not a chore; 
    \vspace{-4pt}
    \item $G_4$ contains the edges that have not been included in $G_1$, $G_2$, $G_3$. 
    Notice that each edge in $G_4$ is a chore for both its endpoint agents. 
    \vspace{-4pt}
\end{itemize}

We will allocate the unallocated edges such that no agent will receive an edge that is a chore for her and thus Property (4) will be retained. 
Besides, no agent will get worse off and no allocated edge will become unallocated, which will ensure that Properties (1) and (2) are retained. 
Moreover, no new envy will occur, which will ensure that Properties (6) and (7) are retained. 
Furthermore, no allocated edge will be reallocated to another agent, which will ensure that an agent who is safe for some agent is always safe for that agent and thus Properties (3) and (8) are retained. 
We will also see that the (partial) allocation is always $\EFX^0_-$ during the allocation process. 
Specifically, 
\begin{itemize} 
    \vspace{-4pt}
    \item For each edge in $G_1$, we allocate it to the non-envied endpoint agent for whom it is not a chore.
    \vspace{-4pt}
    \item For each edge in $G_2$, we allocate it to the non-envied agent who is safe for both its endpoint agents. 
    \vspace{-4pt}
    \item For each edge $e_{i, j}$ in $G_3$, we consider three cases. 
    Without loss of generality, let $a_i$ be the non-envied endpoint agent for whom $e_{i, j}$ is not a chore and $a_j$ be the envied endpoint agent for whom it is a chore. 
    First, $a_i$ becomes non-envied. Similar to the allocation of $G_1$, we allocate the edge to $a_i$. 
    Second, there exists a non-envied agent $a_k \neq a_j$ who is safe for $a_i$. Similar to the allocation of $G_2$, we allocate $e_{i, j}$ to $a_k$. 
    Third, $a_j$ is the only non-envied agent who is safe for $a_i$. 
    By Property (8), it must be the case that there exists a sequence of agents $a_{i_0} \leftarrow a_{i_1} \leftarrow \cdots \leftarrow a_{i_s}$ such that  $a_{i_l}$ envies  $a_{i_{l-1}}$ for every $l \in [s]$, $a_{i_0}$ is $a_i$ and $a_{i_s}$ is $a_j$. 
    For this case, we allocate $e_{i, j}$ to $a_{i_{s-1}}$. 
    \vspace{-4pt}
    \item For $G_4$, we consider two cases. 
    First, if no agent is envied, we allocate each edge to an agent who is not its endpoint. 
    Second, if some agent is envied, we find two agents $a_i$ and $a_j$ such that $a_i$ is envied by $a_j$ and $a_j$ is non-envied. 
    We allocate the edges in $G_4$ that are incident to $a_j$ (i.e., $E_j \cap G_4$) to $a_i$, and the other edges in $G_4$ (i.e., $G_4 \setminus E_j$) to $a_j$.
    \vspace{-4pt}
\end{itemize}

We have the following lemma in Part 2. 

\begin{lemma}\label{lem:allocation:EFX0-:second}
    For any simple graph, given a (partial) $\EFX^0_-$ orientation that satisfies Properties (1)-(8) in Definition \ref{def:allocation:EFX0-:properties}, we can compute an $\EFX^0_-$ allocation in polynomial time. 
\end{lemma}
\begin{proof}
    Each edge in $G_1$ is allocated to the non-envied endpoint agent for whom the edge is not a chore.
    This does not incur new envy, since the other endpoint agent of the edge prefers her own bundle to the edge by Property (1). 
    Each edge in $G_2$ is allocated to a non-envied agent who is safe for both its two endpoint agents. 
    By Property (3), such a non-envied agent exists. 
    Since the non-envied agent is safe for both the two endpoint agents, no new envy occurs, either. 
    Moreover, during the allocation of $G_1$ and $G_2$, no agent receives an edge that is a chore for her and no envied agent receives an edge, thus the new (partial) allocation is still $\EFX^0_-$ and retains Properties (4)-(7). 
    Since no agent gets worse off and no allocated edge becomes unallocated, Properties (1) and (2) are retained. 
    Furthermore, since no edge that was allocated to some agent is reallocated to another agent, an agent who is safe for some agent is always safe for that agent and thus Properties (3) and (8) are retained. 

    For each edge $e_{i, j}$ in $G_3$, without loss of generality, let $a_i$ be the endpoint agent for whom $e_{i, j}$ is not a chore and $a_j$ be the other one for whom $e_{i, j}$ is a chore.
    There are three cases: 
    
    \begin{itemize}
        \vspace{-4pt}
        \item First, $a_i$ now becomes non-envied. For this case, we allocate $e_{i, j}$ to $a_i$. 
        By the same reasoning we made when allocating the edges in $G_1$, the new (partial) allocation is still $\EFX^0_-$ and does not break Properties (1)-(8). 
        \vspace{-4pt}
        \item Second, there exists a non-envied agent $a_k \neq a_j$ who is safe for $a_i$. 
        For this case, similarly to the allocation of $G_2$, we allocate $e_{i, j}$ to $a_k$. 
        Although $a_k$ may not be safe for $a_j$, allocating $e_{i, j}$ to $a_k$ does not make $a_j$ envy $a_k$ since $e_{i, j}$ is a chore for $a_j$. 
        Therefore, by the same reasoning we made when allocating the edges in $G_2$, the new (partial) allocation is still $\EFX^0_-$ and does not break Properties (1)-(8). 
        \vspace{-4pt}
        \item Third, $a_j$ is the only non-envied agent who is safe for $a_i$. 
        Since there is no envy cycle among the agents by Property (7), there exists a sequence of agents $a_{i_0} \leftarrow a_{i_1} \leftarrow \cdots \leftarrow a_{i_s}$ such that  $a_{i_l}$ envies  $a_{i_{l-1}}$ for every $l \in [s]$, $a_{i_s}$ is non-envied and $a_{i_0}$ is $a_i$. 
        Then $a_{i_s}$ must be $a_j$; otherwise by Property (8), $a_{i_s}$ is safe for $a_i$, which contradicts the assumption of this case. 
        For this case, we allocate $e_{i, j}$ to $a_{i_{s-1}}$. 
        Clearly, Property (4) still holds. 
        By Property (8), $a_{i_{s-1}}$ is safe for $a_i$. 
        Besides, since $e_{i, j}$ is a chore for $a_j$, allocating $e_{i, j}$ to $a_{i_{s-1}}$ does not incur new envy. 
        Notice that although $a_{i_{s-1}}$ is envied, allocating $e_{i, j}$ to her does not make the (partial) allocation not $\EFX^0_-$ since she is only envied by $a_j$ by Property (6) and $e_{i, j}$ is a chore for $a_j$. 
        Therefore, the new (partial) allocation is still $\EFX^0_-$ and retains Properties (6) and (7). 
        Since no agent gets worse off and no allocated edge becomes unallocated, Properties (1) and (2) are retained. 
        Furthermore, since no edge that was allocated to some agent is reallocated to another agent, an agent who is safe for some agent is always safe for that agent and thus Properties (3) and (8) are retained. 
        \vspace{-4pt}
    \end{itemize}

    At last, we allocate the edges in $G_4$, each of which is a chore for both its endpoint agents. 
    If no agent is envied in the new (partial) allocation, we allocate each edge in $G_4$ to an agent who is not its endpoint and obtain an envy-free allocation. 
    If some agent is envied, we first find two agents $a_i$ and $a_j$ such that $a_i$ is envied by $a_j$ and $a_j$ is non-envied. 
    We are able to find such two agents since there is no envy cycle among the agents by Property (7). 
    Since no new envy occurs during the allocation of $G_1$, $G_2$ and $G_3$, $a_j$ envied $a_i$ in the (partial) orientation computed in Part 1 and thus $e_{i, j}$ was allocated to $a_i$ in Part 1 and is not in $G_4$. 
    We allocate the edges in $G_4$ that are incident to $a_j$ (i.e., $E_j \cap G_4$) to $a_i$, and the other edges (i.e., $G_4 \setminus E_j$) to $a_j$. 
    Although $a_i$ is envied, she can receive the edges in $E_j \cap G_4$ since she is only envied by $a_j$ according to Property (6) and these edges are chores for their endpoint agents (including $a_j$).
    $a_j$ can receive the edges in $G_1 \setminus E_j$ since she is non-envied and these edges are not incident to her and thus are dummies for her. 
    Therefore, the final allocation is still $\EFX^0_-$. 
    Clearly, the whole allocation process runs in polynomial time and we complete the proof. 
\end{proof}

By Lemmas \ref{lem:allocation:EFX0-:first} and \ref{lem:allocation:EFX0-:second}, it is clear that Theorem \ref{thm:allocation:EFX0-} holds.

\subsection{$\EFX_0^+$ Allocations}
\label{subsec:allocation:EFX+0}
Finally, we study $\EFX_0^+$ and have the following theorem. 
\begin{theorem}\label{thm:allocation:EFX+0}
    For any simple graph, an $\EFX_0^+$ allocation always exists and can be computed in polynomial time. 
\end{theorem}

First recall that for chores instances where each edge is a chore for both its endpoint agents, we can compute an envy-free allocation by allocating each edge to an agent who is not its endpoint. 
Thus in the following, we only consider graphs where there exists an edge that is not a chore for at least one of its endpoint agents. 

To get some intuitions about how to compute an $\EFX_0^+$ allocation, consider the graphs where each edge is a good for at least one of its endpoint agents. 
For these graphs, we can simply allocate each edge to the endpoint agent for whom it is a good. 
For any agent $a_i$, each edge she receives is a good for her, and at most one edge that each other agent $a_j$ receives is a good for her.
After removing the good from $a_j$'s bundle, $a_i$ does not envy $a_j$. 
Thus, the allocation is $\EFX_0^+$. 

The trickier graphs to deal with are those with edges that are not goods for any of their endpoint agents. 
For these graphs, we want to find an agent who can receive all such edges, so that we can simply allocate each remaining edge to one of its endpoint agents as above. 
At the same time, the allocation should be $\EFX_0^+$ for the agent we find. 
When there exists an agent $a_i$ to whom the total value of her incident edges is non-negative (i.e., $v_i(E_i) \ge 0$), we let $a_i$ receive all her incident edges as well as all edges that are not goods for any of their endpoint agents. 
We then allocate each remaining edge to one of its endpoint agents for whom it is a good. 
Since $a_i$ receives all her incident edges whose total value is non-negative, the allocation is $\EFX_0^+$ for her. 

However, when the total value of the incident edges is negative to every agent (i.e., $v_i(E_i) < 0$ for every $a_i$), we cannot simply allocate all incident edges to an agent as above, since the allocation may not be $\EFX_0^+$ for her. 
For this case, we let an agent receive all her incident edges that are not chores for her and allocate her other incident edges to another agent. 
More concretely, we first choose an edge $e_{i, j}$ that is not a chore for $a_i$, breaking the tie by giving priority to the edges that are not chores for one endpoint agent and are chores for the other. 
We then let $a_i$ receive all her incident edges that are not chores for her, as well as all edges that are not incident to her but are not goods for any of their endpoint agents. 
Next, we let $a_j$ receive all her unallocated incident edges that are goods for her, as well as $a_i$'s unallocated incident edges that are not goods for any of their endpoint agents. 
At last, we allocate each remaining unallocated edge to one of its endpoint agents for whom it is a good. 
The formal description of the above allocation process is provided in Algorithm \ref{alg:allocation:EFX+0}. 

\begin{algorithm}[tb]
\caption{Computing an $\EFX_0^+$ Allocation}
\label{alg:allocation:EFX+0}
\KwIn{A mixed instance with a graph where $N$ is the vertex set and $M$ is the edge set. }
\KwOut{An $\EFX^+_0$ allocation $\mathbf{X} = (X_1, \ldots, X_n)$.}
Initiate $X_i \leftarrow \emptyset$ for every $a_i \in N$. \\
\If{there exists $a_i \in N$ such that $v_i(E_i) \ge 0$}{
    $X_i \leftarrow X_i \cup E_i$, $M \leftarrow M \setminus E_i$. \\
    \For{each $e_{k, l} \in M$ that is not a good for either $a_k$ or $a_l$}{
        $X_i \leftarrow X_i \cup \{e_{k, l}\}$, $M \leftarrow M \setminus \{e_{k, l}\}$. \\
    }
    \For{each $e_{k, l} \in M$ that is a good for $a_k$}{
        $X_k \leftarrow X_k \cup \{e_{k, l}\}$, $M \leftarrow M \setminus \{e_{k, l}\}$. \\
    }
}
\Else{ 
    Choose an edge $e_{i, j}$ that is not a chore for $a_i$, breaking the tie by giving priority to the edges that are not chores for one endpoint agent and are chores for the other. \label{line:alg:allocation:EFX+0:tie-breaking}\\
    $X_i \leftarrow X_i \cup E_i^{\ge 0}$, $M \leftarrow M \setminus E_i^{\ge 0}$. \\
    \For{each $e_{k, l} \in M$ with $k, l \neq i$ that is not a good for either $a_k$ or $a_l$}{
        $X_i \leftarrow X_i \cup \{e_{k, l}\}$, $M \leftarrow M \setminus \{e_{k, l}\}$. \label{line:alg:allocation:EFX+0:both-nonnegative} \\
    }
    $X_j \leftarrow X_j \cup (M \cap E_j^{>0})$, $M \leftarrow M \setminus E_j^{>0}$. \\
    \For{each $e_{i, k} \in M$ that is not a good for $a_k$}{
        $X_j \leftarrow X_j \cup \{e_{i, k}\}$, $M \leftarrow M \setminus \{e_{i, k}\}$. \\
    }
    \For{each $e_{k, l} \in M$ that is a good for $a_k$}{
        $X_k \leftarrow X_k \cup \{e_{k, l}\}$, $M \leftarrow M \setminus \{e_{k, l}\}$. \\
    }
}
\Return $\mathbf{X} = (X_1, \ldots, X_n)$. 
\end{algorithm}

Now we are ready to prove Theorem \ref{thm:allocation:EFX+0}. 

\begin{proof}[Proof of Theorem \ref{thm:allocation:EFX+0}]
    Clearly, Algorithm \ref{alg:allocation:EFX+0} runs in polynomial time. 
    It remains to show that the computed allocation is $\EFX_0^+$. 
    We consider two cases: 

    \smallskip
    \noindent \textbf{\underline{Case 1}:} there exists an agent $a_i$ such that $v_i(E_i) \ge 0$.
    
    In this case, no edge remains unallocated at the end of the algorithm. 
    To see this, the incident edges of $a_i$, as well as those that are not incident to $a_i$ and are not goods for any of their endpoint agents, are allocated to $a_i$.
    The remaining edges are goods for at least one of their endpoint agents and are allocated to those endpoint agents. 
    
    For agent $a_i$, she receives a non-negative value since she receives all her incident edges whose total value is non-negative to her and other edges she receives are not incident to her and are dummies for her. 
    Thus, each other agent receives a bundle that has a value of 0 to $a_i$ and $a_i$ does not envy them. 
    For each $a_j \neq a_i$, each edge that she receives is a good for her (if exists) and each other agent receives at most one edge that is a good for her. 
    Thus $a_j$ does not envy others after removing the good from their bundles and the allocation is $\EFX^+_0$. 

    \smallskip
    \noindent \textbf{\underline{Case 2}:} $v_i(E_i) < 0$ for every agent $a_i$.
    
    In this case, no edge remains unallocated at the end of the algorithm, either. 
    To see this, the edges that are incident to $a_i$ and are not chores for $a_i$, as well as those that are not incident to $a_i$ and are not goods for any of their endpoint agents, are allocated to $a_i$. 
    The edges that are incident to $a_i$ but are not goods for any of their endpoint agents are allocated to $a_j$. 
    The remaining edges are goods for at least one of their endpoint agents and are allocated to those endpoint agents. 

    For agent $a_i$, she receives a non-negative value since she only receives her incident edges that are not chores for her and other edges she receives are not incident to her and are dummies for her. 
    Thus, each other agent receives a bundle that has a non-positive value to $a_i$ and $a_i$ does not envy them. 
    For each agent $a_k \neq a_i$ or $a_j$, each edge that she receives is a good for her (if exists) and each other agent receives at most one edge that is a good for her. 
    Therefore, the allocation is $\EFX^+_0$ for $a_k$. 
    It remains to consider agent $a_j$. 
    $a_j$ does not envy each agent $a_k \neq a_i$ or $a_j$, since $a_j$ does not receive any edge that is a chore for her and thus receives a non-negative value, and $a_k$ does not receive any edge that is a good for $a_j$ and thus her bundle has a non-positive value to $a_j$. 
    If the edge $(a_i, a_j)$ is a chore for $a_j$, $a_j$ does not envy $a_i$, either, since $a_i$ does not receive any edge that is a good for $a_j$. 
    If $(a_i, a_j)$ is not a chore for $a_j$, the tie-breaking rule implies that each edge in the graph either is not a chore for any of its endpoint agents, or is a chore for both its endpoint agents. 
    Thus, all $a_j$'s incident edges that are chores for $a_j$ are allocated to $a_i$. 
    Since $v_j(E_j) < 0$, $a_i$'s bundle has a negative value to $a_j$ and $a_j$ does not envy $a_i$. 
    Therefore, the allocation is $\EFX^+_0$. 
\end{proof}

Since any $\EFX_-^0$ or $\EFX_0^+$ allocation is also $\EFX^+_-$, we have the following corollary. 
\begin{corollary}\label{coro:allocation:EFX+-}
    For any simple graph, an $\EFX^+_-$ allocation always exists and can be computed in polynomial time. 
\end{corollary}

\section{Conclusion}
In this paper, we give a complete computational study of EFX allocations on graphs when the items are a mixture of goods and chores.
There are some future directions. 
In our setting, exactly two agents are interested in one common item that is incident to both of them. 
One immediate direction is to study the generalized setting with multi-edges where multiple edges exist between two agents or hypergraphs where more than two agents are interested in one common item.
Another direction is to study the setting where agents are also interested in the edges that are not very far away from them. 
To bypass the hardness results of EFX orientations, we have studied some simple graphs including trees, stars and paths. 
One can also study complex graphs for which the existence of EFX orientations can be determined in polynomial time. 

\newpage
\bibliographystyle{plainnat}
\bibliography{ref}

\begin{thebibliography}{67}
\providecommand{\natexlab}[1]{#1}
\providecommand{\url}[1]{\texttt{#1}}
\expandafter\ifx\csname urlstyle\endcsname\relax
  \providecommand{\doi}[1]{doi: #1}\else
  \providecommand{\doi}{doi: \begingroup \urlstyle{rm}\Url}\fi

\bibitem[Afshinmehr et~al.(2024)Afshinmehr, Danaei, Kazemi, Mehlhorn, and Rathi]{afshinmehr2024efx}
Mahyar Afshinmehr, Alireza Danaei, Mehrafarin Kazemi, Kurt Mehlhorn, and Nidhi Rathi.
\newblock Efx allocations and orientations on bipartite multi-graphs: A complete picture.
\newblock \emph{arXiv preprint arXiv:2410.17002}, 2024.

\bibitem[Akrami and Garg(2024)]{akrami2024breaking}
Hannaneh Akrami and Jugal Garg.
\newblock Breaking the 3/4 barrier for approximate maximin share.
\newblock In \emph{Proceedings of the 2024 Annual ACM-SIAM Symposium on Discrete Algorithms (SODA)}, pages 74--91. SIAM, 2024.

\bibitem[Akrami et~al.(2023)Akrami, Mehlhorn, Seddighin, and Shahkarami]{DBLP:conf/nips/AkramiMSS23}
Hannaneh Akrami, Kurt Mehlhorn, Masoud Seddighin, and Golnoosh Shahkarami.
\newblock Randomized and deterministic maximin-share approximations for fractionally subadditive valuations.
\newblock In \emph{NeurIPS}, 2023.

\bibitem[Aleksandrov and Walsh(2019)]{aleksandrov2019greedy}
Martin Aleksandrov and Toby Walsh.
\newblock Greedy algorithms for fair division of mixed manna.
\newblock \emph{CoRR}, abs/1911.11005, 2019.

\bibitem[Aleksandrov and Walsh(2020)]{DBLP:conf/ki/AleksandrovW20}
Martin Aleksandrov and Toby Walsh.
\newblock Two algorithms for additive and fair division of mixed manna.
\newblock In \emph{Proceedings of the 43rd German Conference on Artificial Intelligence}, pages 3--17, 2020.

\bibitem[Amanatidis et~al.(2021)Amanatidis, Birmpas, Filos{-}Ratsikas, Hollender, and Voudouris]{DBLP:journals/tcs/AmanatidisBFHV21}
Georgios Amanatidis, Georgios Birmpas, Aris Filos{-}Ratsikas, Alexandros Hollender, and Alexandros~A. Voudouris.
\newblock Maximum nash welfare and other stories about {EFX}.
\newblock \emph{Theoretical Computer Science}, 863:\penalty0 69--85, 2021.

\bibitem[Amanatidis et~al.(2024)Amanatidis, Filos-Ratsikas, and Sgouritsa]{amanatidis2024pushing}
Georgios Amanatidis, Aris Filos-Ratsikas, and Alkmini Sgouritsa.
\newblock Pushing the frontier on approximate efx allocations.
\newblock In \emph{Proceedings of the 25th ACM Conference on Economics and Computation}, pages 1268--1286, 2024.

\bibitem[Aziz and Mackenzie(2016)]{DBLP:conf/focs/AzizM16}
Haris Aziz and Simon Mackenzie.
\newblock A discrete and bounded envy-free cake cutting protocol for any number of agents.
\newblock In \emph{Proceedings of the 57th IEEE Symposium on Foundations of Computer Science}, pages 416--427, 2016.

\bibitem[Aziz et~al.(2015)Aziz, Gaspers, Mackenzie, and Walsh]{aziz2015fair}
Haris Aziz, Serge Gaspers, Simon Mackenzie, and Toby Walsh.
\newblock Fair assignment of indivisible objects under ordinal preferences.
\newblock \emph{Artificial Intelligence}, 227:\penalty0 71--92, 2015.

\bibitem[Aziz et~al.(2020)Aziz, Moulin, and Sandomirskiy]{aziz2020polynomial}
Haris Aziz, Herv{\'{e}} Moulin, and Fedor Sandomirskiy.
\newblock A polynomial-time algorithm for computing a pareto optimal and almost proportional allocation.
\newblock \emph{Operations Research Letters}, 48\penalty0 (5):\penalty0 573--578, 2020.

\bibitem[Aziz et~al.(2022)Aziz, Caragiannis, Igarashi, and Walsh]{DBLP:journals/aamas/AzizCIW22}
Haris Aziz, Ioannis Caragiannis, Ayumi Igarashi, and Toby Walsh.
\newblock Fair allocation of indivisible goods and chores.
\newblock \emph{Autonomous Agents and Multi-Agent Systems}, 36\penalty0 (1):\penalty0 3, 2022.

\bibitem[Babaioff et~al.(2021)Babaioff, Ezra, and Feige]{DBLP:conf/aaai/BabaioffEF21}
Moshe Babaioff, Tomer Ezra, and Uriel Feige.
\newblock Fair and truthful mechanisms for dichotomous valuations.
\newblock In \emph{Proceedings of the 35th AAAI Conference on Artificial Intelligence}, pages 5119--5126, 2021.

\bibitem[Barman et~al.(2023)Barman, Khan, Shyam, and Sreenivas]{DBLP:conf/sigecom/Barman0SS23}
Siddharth Barman, Arindam Khan, Sudarshan Shyam, and K.~V.~N. Sreenivas.
\newblock Guaranteeing envy-freeness under generalized assignment constraints.
\newblock In \emph{{EC}}, pages 242--269. {ACM}, 2023.

\bibitem[Bei et~al.(2022)Bei, Igarashi, Lu, and Suksompong]{DBLP:journals/siamdm/BeiILS22}
Xiaohui Bei, Ayumi Igarashi, Xinhang Lu, and Warut Suksompong.
\newblock The price of connectivity in fair division.
\newblock \emph{{SIAM} J. Discret. Math.}, 36\penalty0 (2):\penalty0 1156--1186, 2022.

\bibitem[B{\'{e}}rczi et~al.(2020)B{\'{e}}rczi, B{\'{e}}rczi{-}Kov{\'{a}}cs, Boros, Gedefa, Kamiyama, Kavitha, Kobayashi, and Makino]{DBLP:journals/corr/abs-2006-04428}
Krist{\'{o}}f B{\'{e}}rczi, Erika~R. B{\'{e}}rczi{-}Kov{\'{a}}cs, Endre Boros, Fekadu~Tolessa Gedefa, Naoyuki Kamiyama, Telikepalli Kavitha, Yusuke Kobayashi, and Kazuhisa Makino.
\newblock Envy-free relaxations for goods, chores, and mixed items.
\newblock \emph{CoRR}, abs/2006.04428, 2020.

\bibitem[Berman et~al.(2007)Berman, Karpinski, and Scott]{DBLP:journals/dam/BermanKS07}
Piotr Berman, Marek Karpinski, and Alexander~D. Scott.
\newblock Computational complexity of some restricted instances of 3-sat.
\newblock \emph{Discrete Applied Mathematics}, 155\penalty0 (5):\penalty0 649--653, 2007.

\bibitem[Bhaskar and Pandit(2024)]{bhaskar2024efx}
Umang Bhaskar and Yeshwant Pandit.
\newblock Efx allocations on some multi-graph classes.
\newblock \emph{arXiv preprint arXiv:2412.06513}, 2024.

\bibitem[Bhaskar et~al.(2021)Bhaskar, Sricharan, and Vaish]{bhaskar2021approximate}
Umang Bhaskar, AR~Sricharan, and Rohit Vaish.
\newblock On approximate envy-freeness for indivisible chores and mixed resources.
\newblock \emph{Approximation, Randomization, and Combinatorial Optimization. Algorithms and Techniques}, 207:\penalty0 1:1--1:23, 2021.

\bibitem[Bil{\`{o}} et~al.(2022)Bil{\`{o}}, Caragiannis, Flammini, Igarashi, Monaco, Peters, Vinci, and Zwicker]{DBLP:journals/geb/BiloCFIMPVZ22}
Vittorio Bil{\`{o}}, Ioannis Caragiannis, Michele Flammini, Ayumi Igarashi, Gianpiero Monaco, Dominik Peters, Cosimo Vinci, and William~S. Zwicker.
\newblock Almost envy-free allocations with connected bundles.
\newblock \emph{Games Econ. Behav.}, 131:\penalty0 197--221, 2022.

\bibitem[Biswas and Barman(2018)]{DBLP:conf/ijcai/BiswasB18}
Arpita Biswas and Siddharth Barman.
\newblock Fair division under cardinality constraints.
\newblock In \emph{{IJCAI}}, pages 91--97. ijcai.org, 2018.

\bibitem[Bouveret et~al.(2017)Bouveret, Cechl{\'{a}}rov{\'{a}}, Elkind, Igarashi, and Peters]{DBLP:conf/ijcai/BouveretCEIP17}
Sylvain Bouveret, Katar{\'{\i}}na Cechl{\'{a}}rov{\'{a}}, Edith Elkind, Ayumi Igarashi, and Dominik Peters.
\newblock Fair division of a graph.
\newblock In \emph{Proceedings of the 26th International Joint Conference on Artificial Intelligence}, pages 135--141, 2017.

\bibitem[Bouveret et~al.(2019)Bouveret, Cechl{\'{a}}rov{\'{a}}, and Lesca]{DBLP:journals/aamas/BouveretCL19}
Sylvain Bouveret, Katar{\'{\i}}na Cechl{\'{a}}rov{\'{a}}, and Julien Lesca.
\newblock Chore division on a graph.
\newblock \emph{Auton. Agents Multi Agent Syst.}, 33\penalty0 (5):\penalty0 540--563, 2019.

\bibitem[Budish(2011)]{budish2011combinatorial}
Eric Budish.
\newblock The combinatorial assignment problem: Approximate competitive equilibrium from equal incomes.
\newblock \emph{Journal of Political Economy}, 119\penalty0 (6):\penalty0 1061--1103, 2011.

\bibitem[Caragiannis et~al.(2019{\natexlab{a}})Caragiannis, Gravin, and Huang]{DBLP:conf/ec/CaragiannisGH19}
Ioannis Caragiannis, Nick Gravin, and Xin Huang.
\newblock Envy-freeness up to any item with high nash welfare: The virtue of donating items.
\newblock In \emph{Proceedings of the 20th ACM Conference on Economics and Computation}, pages 527--545, 2019{\natexlab{a}}.

\bibitem[Caragiannis et~al.(2019{\natexlab{b}})Caragiannis, Kurokawa, Moulin, Procaccia, Shah, and Wang]{DBLP:journals/teco/CaragiannisKMPS19}
Ioannis Caragiannis, David Kurokawa, Herv{\'{e}} Moulin, Ariel~D. Procaccia, Nisarg Shah, and Junxing Wang.
\newblock The unreasonable fairness of maximum nash welfare.
\newblock \emph{{ACM} Trans. Economics and Comput.}, 7\penalty0 (3):\penalty0 12:1--12:32, 2019{\natexlab{b}}.

\bibitem[Chandramouleeswaran et~al.(2024)Chandramouleeswaran, Nimbhorkar, and Rathi]{chandramouleeswaran2024fair}
Harish Chandramouleeswaran, Prajakta Nimbhorkar, and Nidhi Rathi.
\newblock Fair division in a variable setting.
\newblock \emph{arXiv preprint arXiv:2410.14421}, 2024.

\bibitem[Chaudhury et~al.(2020)Chaudhury, Garg, and Mehlhorn]{chaudhury2020efx}
Bhaskar~Ray Chaudhury, Jugal Garg, and Kurt Mehlhorn.
\newblock Efx exists for three agents.
\newblock In \emph{Proceedings of the 21st ACM Conference on Economics and Computation}, pages 1--19, 2020.

\bibitem[Christodoulou et~al.(2023)Christodoulou, Fiat, Koutsoupias, and Sgouritsa]{christodoulou2023fair}
George Christodoulou, Amos Fiat, Elias Koutsoupias, and Alkmini Sgouritsa.
\newblock Fair allocation in graphs.
\newblock In \emph{Proceedings of the 24th ACM Conference on Economics and Computation}, pages 473--488, 2023.

\bibitem[Dehghani et~al.(2018)Dehghani, Farhadi, Hajiaghayi, and Yami]{DBLP:conf/soda/DehghaniFHY18}
Sina Dehghani, Alireza Farhadi, Mohammad~Taghi Hajiaghayi, and Hadi Yami.
\newblock Envy-free chore division for an arbitrary number of agents.
\newblock In \emph{Proceedings of the 29th Annual ACM-SIAM Symposium on Discrete Algorithms}, pages 2564--2583, 2018.

\bibitem[Deligkas et~al.(2024)Deligkas, Eiben, Goldsmith, and Korchemna]{deligkas2024ef1}
Argyrios Deligkas, Eduard Eiben, Tiger-Lily Goldsmith, and Viktoriia Korchemna.
\newblock Ef1 and efx orientations.
\newblock \emph{arXiv preprint arXiv:2409.13616}, 2024.

\bibitem[Elkind et~al.(2024)Elkind, Igarashi, and Teh]{DBLP:conf/sagt/ElkindIT24}
Edith Elkind, Ayumi Igarashi, and Nicholas Teh.
\newblock Fair division of chores with budget constraints.
\newblock In \emph{{SAGT}}, volume 15156 of \emph{Lecture Notes in Computer Science}, pages 55--71. Springer, 2024.

\bibitem[Gafni et~al.(2023)Gafni, Huang, Lavi, and Talgam-Cohen]{gafni2023unified}
Yotam Gafni, Xin Huang, Ron Lavi, and Inbal Talgam-Cohen.
\newblock Unified fair allocation of goods and chores via copies.
\newblock \emph{ACM Transactions on Economics and Computation}, 11\penalty0 (3-4):\penalty0 1--27, 2023.

\bibitem[Gamow and Stern(1958)]{GS58}
George Gamow and Marvin Stern.
\newblock \emph{Puzzle-Math}.
\newblock Viking press, 1958.

\bibitem[Ghodsi et~al.(2022)Ghodsi, Hajiaghayi, Seddighin, Seddighin, and Yami]{DBLP:journals/ai/GhodsiHSSY22}
Mohammad Ghodsi, Mohammad~Taghi Hajiaghayi, Masoud Seddighin, Saeed Seddighin, and Hadi Yami.
\newblock Fair allocation of indivisible goods: Beyond additive valuations.
\newblock \emph{Artif. Intell.}, 303:\penalty0 103633, 2022.

\bibitem[Hosseini et~al.(2020)Hosseini, Sikdar, Vaish, Wang, and Xia]{hosseini2020fair}
Hadi Hosseini, Sujoy Sikdar, Rohit Vaish, Hejun Wang, and Lirong Xia.
\newblock Fair division through information withholding.
\newblock In \emph{Proceedings of the AAAI Conference on Artificial Intelligence}, volume~34, pages 2014--2021, 2020.

\bibitem[Hosseini et~al.(2021)Hosseini, Sikdar, Vaish, and Xia]{DBLP:conf/aaai/HosseiniSVX21}
Hadi Hosseini, Sujoy Sikdar, Rohit Vaish, and Lirong Xia.
\newblock Fair and efficient allocations under lexicographic preferences.
\newblock In \emph{Proceedings of the 35th AAAI Conference on Artificial Intelligence}, pages 5472--5480, 2021.

\bibitem[Hsu(2024)]{hsu2024efx}
Kevin Hsu.
\newblock Efx orientations of multigraphs.
\newblock \emph{arXiv preprint arXiv:2410.12039}, 2024.

\bibitem[Hsu and King(2025)]{hsu2025polynomial}
Kevin Hsu and Valerie King.
\newblock A polynomial-time algorithm for efx orientations of chores.
\newblock \emph{arXiv preprint arXiv:2501.13481}, 2025.

\bibitem[Hummel(2024)]{DBLP:journals/corr/abs-2404-11582}
Halvard Hummel.
\newblock Maximin shares in hereditary set systems.
\newblock \emph{CoRR}, abs/2404.11582, 2024.

\bibitem[Hummel and Hetland(2022{\natexlab{a}})]{DBLP:conf/eumas/HummelH22}
Halvard Hummel and Magnus~Lie Hetland.
\newblock Maximin shares under cardinality constraints.
\newblock In \emph{{EUMAS}}, volume 13442 of \emph{Lecture Notes in Computer Science}, pages 188--206. Springer, 2022{\natexlab{a}}.

\bibitem[Hummel and Hetland(2022{\natexlab{b}})]{DBLP:journals/aamas/HummelH22}
Halvard Hummel and Magnus~Lie Hetland.
\newblock Fair allocation of conflicting items.
\newblock \emph{Autonomous Agents and Multi-Agent Systems}, 36\penalty0 (1):\penalty0 8, 2022{\natexlab{b}}.

\bibitem[Igarashi and Peters(2019)]{DBLP:conf/aaai/IgarashiP19}
Ayumi Igarashi and Dominik Peters.
\newblock Pareto-optimal allocation of indivisible goods with connectivity constraints.
\newblock In \emph{{AAAI}}, pages 2045--2052. {AAAI} Press, 2019.

\bibitem[Karp et~al.(1975)Karp, Miller, and Thatcher]{karp1975reducibility}
Richard~M Karp, Raymond~E Miller, and James~W Thatcher.
\newblock Reducibility among combinatorial problems.
\newblock \emph{Journal of Symbolic Logic}, 40\penalty0 (4), 1975.

\bibitem[Kaviani et~al.(2024)Kaviani, Seddighin, and Shahrezaei]{kaviani2024almost}
Alireza Kaviani, Masoud Seddighin, and AmirMohammad Shahrezaei.
\newblock Almost envy-free allocation of indivisible goods: A tale of two valuations.
\newblock \emph{arXiv preprint arXiv:2407.05139}, 2024.

\bibitem[Kulkarni et~al.(2021)Kulkarni, Mehta, and Taki]{KulkarniMT21}
Rucha Kulkarni, Ruta Mehta, and Setareh Taki.
\newblock Indivisible mixed manna: On the computability of {MMS+PO} allocations.
\newblock In \emph{Proceedings of the 22nd {ACM} Conference on Economics and Computation}, pages 683--684, 2021.

\bibitem[Kumar et~al.(2024)Kumar, Equbal, Gurjar, Nath, and Vaish]{DBLP:conf/atal/KumarEGNV24}
Yatharth Kumar, Sarfaraz Equbal, Rohit Gurjar, Swaprava Nath, and Rohit Vaish.
\newblock Fair scheduling of indivisible chores.
\newblock In \emph{{AAMAS}}, pages 2345--2347. International Foundation for Autonomous Agents and Multiagent Systems / {ACM}, 2024.

\bibitem[Li et~al.(2021)Li, Li, and Zhang]{DBLP:conf/nips/LiLZ21}
Bo~Li, Minming Li, and Ruilong Zhang.
\newblock Fair scheduling for time-dependent resources.
\newblock In \emph{NeurIPS}, pages 21744--21756, 2021.

\bibitem[Li et~al.(2022)Li, Li, and Wu]{DBLP:conf/www/0037L022}
Bo~Li, Yingkai Li, and Xiaowei Wu.
\newblock Almost (weighted) proportional allocations for indivisible chores.
\newblock In \emph{Proceedings of the ACM Web Conference 2022}, pages 122--131, 2022.

\bibitem[Li and Vetta(2021)]{DBLP:journals/teco/LiV21}
Zhentao Li and Adrian Vetta.
\newblock The fair division of hereditary set systems.
\newblock \emph{{ACM} Trans. Economics and Comput.}, 9\penalty0 (2):\penalty0 12:1--12:19, 2021.

\bibitem[Lindner and Rothe(2016)]{LindnerR16}
Claudia Lindner and J{\"{o}}rg Rothe.
\newblock Cake-cutting: Fair division of divisible goods.
\newblock In J{\"{o}}rg Rothe, editor, \emph{Economics and Computation}, chapter~7, pages 395--491. Springer, 2016.

\bibitem[Lipton et~al.(2004)Lipton, Markakis, Mossel, and Saberi]{lipton2004approximately}
Richard~J Lipton, Evangelos Markakis, Elchanan Mossel, and Amin Saberi.
\newblock On approximately fair allocations of indivisible goods.
\newblock In \emph{Proceedings of the 5th ACM Conference on Electronic Commerce}, pages 125--131, 2004.

\bibitem[Liu et~al.(2023)Liu, Lu, Suzuki, and Walsh]{DBLP:journals/corr/abs-2306-09564}
Shengxin Liu, Xinhang Lu, Mashbat Suzuki, and Toby Walsh.
\newblock Mixed fair division: {A} survey.
\newblock \emph{CoRR}, abs/2306.09564, 2023.

\bibitem[Madathil(2023)]{DBLP:conf/ijcai/Madathil23}
Jayakrishnan Madathil.
\newblock Fair division of a graph into compact bundles.
\newblock In \emph{Proceedings of the 32nd International Joint Conference on Artificial Intelligence}, pages 2835--2843, 2023.

\bibitem[Mahara(2021)]{DBLP:conf/esa/Mahara21}
Ryoga Mahara.
\newblock Extension of additive valuations to general valuations on the existence of {EFX}.
\newblock In \emph{Proceedings of the 29th European Symposium on Algorithms}, pages 1--15, 2021.

\bibitem[Misra and Sethia(2024)]{misra2024envy}
Neeldhara Misra and Aditi Sethia.
\newblock Envy-free and efficient allocations for graphical valuations.
\newblock In \emph{International Conference on Algorithmic Decision Theory}, pages 258--272. Springer, 2024.

\bibitem[Moulin(2003)]{books/daglib/0017734}
Herv{\'{e}} Moulin.
\newblock \emph{Fair division and collective welfare}.
\newblock {MIT} Press, 2003.

\bibitem[Payan et~al.(2023)Payan, Sengupta, and Viswanathan]{DBLP:conf/atal/PayanSV23}
Justin Payan, Rik Sengupta, and Vignesh Viswanathan.
\newblock Relaxations of envy-freeness over graphs.
\newblock In \emph{Proceedings of the 2023 International Conference on Autonomous Agents and Multiagent Systems}, pages 2652--2654, 2023.

\bibitem[Plaut and Roughgarden(2020)]{plaut2020almost}
Benjamin Plaut and Tim Roughgarden.
\newblock Almost envy-freeness with general valuations.
\newblock \emph{SIAM Journal on Discrete Mathematics}, 34\penalty0 (2):\penalty0 1039--1068, 2020.

\bibitem[Procaccia(2016)]{Procaccia_cake_16}
Ariel~D. Procaccia.
\newblock Cake cutting algorithms.
\newblock In Felix Brandt, Vincent Conitzer, Ulle Endriss, J{\'{e}}r{\^{o}}me Lang, and Ariel~D. Procaccia, editors, \emph{Handbook of Computational Social Choice}, chapter~13, pages 311--330. Cambridge University Press, 2016.

\bibitem[Seddighin and Seddighin(2024)]{DBLP:journals/ai/SeddighinS24}
Masoud Seddighin and Saeed Seddighin.
\newblock Improved maximin guarantees for subadditive and fractionally subadditive fair allocation problem.
\newblock \emph{Artif. Intell.}, 327:\penalty0 104049, 2024.

\bibitem[Steinhaus(1949)]{Steinhaus49}
Hugo Steinhaus.
\newblock Sur la division pragmatique.
\newblock \emph{Econometrica}, 17 (Supplement):\penalty0 315--319, 1949.

\bibitem[Suksompong(2019)]{DBLP:journals/dam/Suksompong19}
Warut Suksompong.
\newblock Fairly allocating contiguous blocks of indivisible items.
\newblock \emph{Discret. Appl. Math.}, 260:\penalty0 227--236, 2019.

\bibitem[Suksompong(2021)]{suksompong2021constraints}
Warut Suksompong.
\newblock Constraints in fair division.
\newblock \emph{ACM SIGecom Exchanges}, 19\penalty0 (2):\penalty0 46--61, 2021.

\bibitem[Varian(1974)]{Varian74}
Hal~R. Varian.
\newblock Equity, envy and efficiency.
\newblock \emph{Journal of Economic Theory}, 9:\penalty0 63--91, 1974.

\bibitem[Wu et~al.(2025)Wu, Li, and Gan]{DBLP:journals/iandc/WuLG25}
Xiaowei Wu, Bo~Li, and Jiarui Gan.
\newblock Approximate envy-freeness in indivisible resource allocation with budget constraints.
\newblock \emph{Inf. Comput.}, 303:\penalty0 105264, 2025.

\bibitem[Zeng and Mehta(2024)]{zeng2024structure}
Jinghan~A Zeng and Ruta Mehta.
\newblock On the structure of envy-free orientations on graphs.
\newblock \emph{arXiv preprint arXiv:2404.13527}, 2024.

\bibitem[Zhou et~al.(2024)Zhou, Wei, Li, and Li]{DBLP:conf/ijcai/0027WL024}
Yu~Zhou, Tianze Wei, Minming Li, and Bo~Li.
\newblock A complete landscape of {EFX} allocations on graphs: Goods, chores and mixed manna.
\newblock In \emph{{IJCAI}}, pages 3049--3056. ijcai.org, 2024.

\end{thebibliography}

\newpage
\appendix
\section*{Appendix}

\section{Computing a (Partial) $\EFX^0_{-}$ Orientation that Satisfies Properties (1)-(8) in Definition \ref{def:allocation:EFX0-:properties}}
\label{ap:allocation:EFX0-:part1}
In this section, we show how to compute a (partial) $\EFX^0_{-}$ orientation that satisfies Properties (1)-(8) in Definition \ref{def:allocation:EFX0-:properties}. 
We achieve this in two steps. 

\subsection{Step 1: Computing an Initial (Partial) Orientation $\mathbf{X}^1$}

We first construct an initial (partial) orientation $\mathbf{X}^1$ by letting each agent pick one incident edge (if it exists) that she values the most among those that have not been allocated and are not chores for her. 
Once an agent picks an edge, the other endpoint agent of the edge will be the next to pick an edge.
The formal description is provided in Algorithm \ref{alg:allocation:EFX0-:initial}. 
It is easy to see that the following claim holds. 

\begin{claim}\label{clm:allocation:EFX0-:initial}
    The initial (partial) orientation $\mathbf{X}^1$ computed in Step 1 is $\EFX^0_-$ and satisfies Properties (1) and (4)-(7) in Definition \ref{def:allocation:EFX0-:properties}. 
    Moreover, it is computed in polynomial time. 
\end{claim}
\begin{proof}
    Clearly, Algorithm \ref{alg:allocation:EFX0-:initial} runs in polynomial time. 
    Since each agent receives at most one incident edge that is not a chore for her, $\mathbf{X}^1$ is $\EFX^0_-$ and satisfies Properties (4)-(6). 
    Since each agent picks the edge that she values the most among those that have not been allocated and are not chores for her, $\mathbf{X}^1$ satisfies Properties (1) and (7). 
\end{proof}

\begin{algorithm}[!tb]
\caption{Computing an Initial Orientation}
\label{alg:allocation:EFX0-:initial}
\KwIn{A mixed instance with a graph where $N$ is the vertex set and $M$ is the edge set.}
\KwOut{An initial (partial) orientation $\mathbf{X}^1 = (X_1^1, \ldots, X_n^1)$.}
Initialize $X_i^1 \leftarrow \emptyset$ for every $a_i \in N$. \\
\While{there exists $a_i \in N$ such that $M \cap E_i^{\ge 0} \neq \emptyset$ and $X_i^1 = \emptyset$}{
    Let $e_{i, j}$ be the edge in $M \cap E_i^{\ge 0}$ that $a_i$ values the most, i.e., $e_{i, j} \in \arg\max_{e \in M \cap E_i^{\ge 0}}v_i(e)$.\\
    $X_i^1 \leftarrow \{e_{i, j}\}$, $M \leftarrow M \setminus \{e_{i, j}\}$. \\
    Let $k \leftarrow j$. \\
    \While{$M \cap E_k^{\ge 0} \neq \emptyset$ and $X_k^1 = \emptyset$}{
        Let $e_{k, j }$ be the edge in $M \cap E_k^{\ge 0}$ that $a_k$ values the most. \\
        $X_k^1 \leftarrow \{e_{k, j}\}$, $M \leftarrow M \setminus \{e_{k, j}\}$. \\
        $k \leftarrow j$. \\
    }
}
\Return $\mathbf{X}^1 = (X_1^1, \ldots, X_n^1)$. \\
\end{algorithm}

After Step 1, either the following Step 2-1 or Step 2-2 will be executed, depending on whether $\mathbf{X}^1$ satisfies Properties (2) and (3). 

\subsection{Step 2-1: Computing the Desired (Partial) Orientation if $\mathbf{X}^1$ does not Satisfy Property (2)}
\vspace{0.5mm}

Since $\mathbf{X}^1$ does not satisfy Property (2), there exists an envied agent $a_i$, for whom her bundle is less valuable than the unallocated incident edges that are not chores for her (i.e., $v_i(X_i^1) < v_i(E_i^{\ge 0} \cap R(\mathbf{X}^1)))$. 
We allocate all edges in $E_i^{\ge 0} \cap R(\mathbf{X}^1)$ to $a_i$ and release the edge she received in $\mathbf{X}^1$, which may break Property (1). 
In order to regain Property (1), for any agent who prefers an unallocated incident edge to her current bundle, we allocate the unallocated edge to the agent and release all edges she received previously. 
We repeat the above process until Property (2) is satisfied. 
The formal description is provided in Algorithm \ref{alg:allocation:EFX0-:property2}. 
We have the following claim. 

\begin{algorithm}[tb]
\caption{Satisfying Property (2)}
\label{alg:allocation:EFX0-:property2}
\KwIn{A (partial) orientation $\mathbf{X} = (X_1, \ldots, X_n)$ that does not satisfy Property (2).}
\KwOut{A new (partial) orientation $\mathbf{X}^{\prime} = (X_1^{\prime}, \ldots, X_n^{\prime})$ that satisfies Property (2).}
Initialize $X_i^{\prime} \leftarrow X_i$ for every $a_i \in N$. \\
\While{there exists $a_i \in N$ such that $a_i$ is envied and $v_i(X_i^{\prime}) < v_i(E_i^{\ge 0} \cap R(\mathbf{X}^{\prime}))$}{
    $S \leftarrow X_i^{\prime}$. \\
    $X_i^{\prime} \leftarrow E_i^{\ge 0} \cap R(\mathbf{X}^{\prime})$.\\
    Release all edges in $S$. \\
    \While{there exist $a_j \in N$ and $e \in E_j^{\ge 0} \cap R(\mathbf{X}^{\prime})$ such that $v_j(X_j^{\prime}) < v_j(e)$}{
        $S \leftarrow X_j^{\prime}$. \\
        $X_j^{\prime} \leftarrow \{e\}$. \\
        Release all edges in $S$. \\
    }
}
\Return $\mathbf{X}^{\prime} = (X_1^{\prime}, \ldots, X_n^{\prime})$. \\
\end{algorithm}

\begin{claim}\label{clm:allocation:EFX0-:satisfy2}
    Taking as input a (partial) orientation that is $\EFX^0_-$ and satisfies Properties (1) and (4)-(7) in Definition \ref{def:allocation:EFX0-:properties}, Algorithm \ref{alg:allocation:EFX0-:property2} computes a (partial) orientation that is $\EFX^0_-$ and satisfies Properties (1)-(7) in polynomial time. 
\end{claim}
\begin{proof}
    Clearly, no agent receives an edge that is a chore for her in Algorithm \ref{alg:allocation:EFX0-:property2}, thus Property (4) always holds. 
    We next show by induction that Properties (1) and (5)-(7) always hold at the end of each round of the outer while-loop.
    Assume that at the beginning of a round of the outer while-loop, these four properties hold. 
    Let $a_i$ be the agent who makes the algorithm enter the outer while-loop and $a_j$ be the agent who envies $a_i$. 
    By Property (5), $a_i$ only owned $e_{i, j}$ at the beginning of the round. 
    In the outer while-loop, $a_i$ receives all edges in $E_i^{\ge 0} \cap R(\mathbf{X}^{\prime})$ and releases $e_{i, j}$. 
    Since the value of $E_i^{\ge 0} \cap R(\mathbf{X}^{\prime})$ to $a_i$ is larger than that of $e_{i, j}$, Property (1) holds for $a_i$. 
    Besides, $a_i$ becomes non-envied, since the other endpoint agents of the edges in $E_i^{\ge 0} \cap R(\mathbf{X}^{\prime})$ do not envy her by Property (1). 
    Since $a_i$ gets better off, she does not envy the agents whom she did not envy previously. 
    Therefore, the new (partial) orientation retains Properties (5)-(7) after $a_i$ gets her new bundle. 
    
    After $a_i$ releases $e_{i, j}$, Property (1) does not hold for $a_j$, and the algorithm enters the inner while-loop. 
    In the inner while-loop, $a_j$ receives $e_{i, j}$ and releases the edges she owned previously. 
    The value of $e_{i, j}$ to $a_j$ is larger than the total value of the edges $a_j$ owned, and is larger than the value of each of these edges since none of them are chores for $a_j$ by Property (4). 
    Furthermore, since Property (1) held for $a_j$ at the beginning of the round of outer while-loop, the value of $e_{i, j}$ to $a_j$ is larger than the value of each of $a_j$'s incident edges that were unallocated at the beginning of the round and are not chores for her. 
    Thus, Property (1) holds for $a_j$. 
    Besides, no one envies $a_j$ since she only receives the edge $e_{i, j}$ and $a_i$ now prefers her own bundle to $e_{i, j}$. 
    Since $a_j$ gets better off, she does not envy the agents whom she did not envy previously. 
    Therefore, the new (partial) orientation retains Properties (5)-(7) after $a_j$ gets her new bundle. 
        
    If $a_j$ was non-envied at the beginning of the round of the outer while-loop, each agent now prefers her own bundle to the edge between herself and $a_j$ that $a_j$ releases. 
    Therefore, the round ends and Properties (1) and (5)-(7) hold. 
    If $a_j$ was envied by some agent, Property (1) is not satisfied for that agent and the algorithm enters the inner while-loop again. 
    By the same reasoning, Properties (1) and (5)-(7) hold after that round of the inner while-loop. 
    By induction, Properties (1) and (5)-(7) hold at the end of the the round of the outer while-loop. 
    Therefore, by induction, Properties (1) and (5)-(7) hold at the end of the algorithm. 

    Next, we show that Properties (2) and (3) hold at the end of the algorithm. 
    Observe that the agents who receive new bundles during the algorithm become non-envied, and the agents who were non-envied previously are still non-envied since no new envy occurs during the algorithm. 
    Therefore, the algorithm can terminate in polynomial time, and we have Property (2) at the end. 
    To see Property (3), let $a_k$ be the agent who makes the algorithm enter the last round of the outer while-loop and $a_l$ be the agent involved in the first round of the inner while-loop in that round. 
    At the end of that round, both $a_k$ and $a_l$ are non-envied, and $a_l$ only receives $e_{k, l}$. 
    By Property (2), $a_l$ is safe for all envied agents and Property (3) holds. 

    The (partial) orientation is $\EFX^0_-$ since each agent only receives the edges that are not chores for her by Property (4), and each envied agent only receives one edge by Property (5). 
    Therefore, we complete the proof. 
\end{proof}

\subsection{Step 2-2: Computing the Desired (Partial) Orientation if $\mathbf{X}^1$ Satisfies Property (2) but (3)}

Since $\mathbf{X}^1$ does not satisfy Property (3), there exist two envied agents $a_i$ and $a_j$ such that no non-envied agent is safe for both of them. 
Let $a_{i_0} \leftarrow a_{i_1} \leftarrow \cdots \leftarrow a_{i_s}$ be the sequence of agents such that $a_{i_0}$ is $a_i$, $a_{i_l}$ envies $a_{i_{l-1}}$ for every $l \in [s]$ and $a_{i_s}$ is non-envied, which exists since there is no envy cycle among the agents by Property (7). 
Similarly, let $a_{j_0} \leftarrow  a_{j_1} \leftarrow \cdots \leftarrow a_{j_t}$ be the sequence of agents such that $a_{j_0}$ is $a_j$, $a_{j_l}$ envies $a_{j_{l-1}}$ for every $l \in [t]$ and $a_{j_t}$ is non-envied. 
We first consider the following case: 

\medskip
\noindent\textbf{\underline{Case 1}:} Edge $e_{i, i_s}$ exists and $e_{i, i_s}$ is allocated to $a_{i_s}$ in $\mathbf{X}^1$. 

Since each agent receives at most one edge in $\mathbf{X}^1$, $a_{i_s}$ only receives $e_{i, i_s}$ and does not receive any incident edge of $a_j$. 
Thus, each edge that $a_{i_s}$ receives is a dummy for $a_j$ and by Property (2), $a_{i_s}$ is safe for $a_j$. 
According to our assumption, $a_{i_s}$ is not safe for $a_i$; that is, $v_i(\{e_{i, i_s}\} \cup (E_i^{\ge 0} \cap R(\mathbf{X}^1))) > v_i(X_i^1)$. 
For this case, we allocate the edges in $E_i^{\ge 0} \cap R(\mathbf{X}^1)$ as well as $e_{i, i_s}$ to $a_i$, and $e_{i_{l-1}, i_l}$ to $a_{i_l}$ for every $l \in [s]$. 
Figure \ref{fig:allocation:EFX0-:step2-2:case1} visualizes the allocation process. 
We have the following claim. 

\begin{figure}[tb]
\centering
\includegraphics[width=0.45\columnwidth]{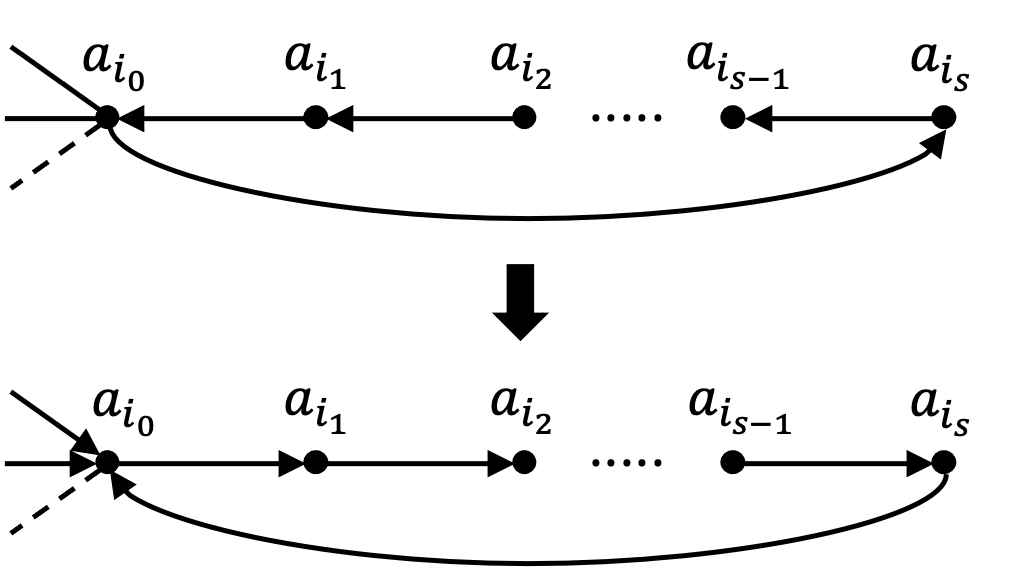}
\caption{The allocation process for Case 1 of Step 2-2. Note that there is a dashed edge incident to $a_{i_0}$, which is a chore for $a_{i_0}$ and is not allocated to $a_{i_0}$ during the process. 
Note that for figures in this section, an arrow $a\rightarrow b$ means that the edge between $a$ and $b$ is allocated to $b$.} 
\label{fig:allocation:EFX0-:step2-2:case1}
\end{figure}

\begin{claim}
\label{clm:allocation:EFX0-:step2-2:case1}
    The (partial) orientation computed in Case 1 of Step 2-2 is $\EFX^0_-$ and satisfies Properties (1)-(7) in Definition \ref{def:allocation:EFX0-:properties}. 
    Moreover, it is computed in polynomial time. 
\end{claim}
\begin{proof}
    Clearly, the allocation process runs in polynomial time and Property (4) holds.  
    We then show that the (partial) orientation satisfies Properties (1), (2) and (5)-(7). 
    Observe that the agents who receive new bundles in the allocation process become non-envied. 
    To see this, first consider $a_i$. 
    Since $\mathbf{X}^1$ satisfies Property (1), the other endpoint agents of the edges in $E_i^{\ge 0} \cap R(\mathbf{X}^1)$ do not envy $a_i$. 
    Besides, $a_{i_s}$ does not envy $a_i$, since she envied $a_{i_{s-1}}$ previously and prefers $e_{i_{s-1}, i_s}$ that she owns currently to $e_{i, i_s}$. 
    For each $l \in [s]$, $a_{i_l}$ is non-envied, since $a_{i_{l-1}}$ prefers her current bundle to $e_{i_{l-1}, i_l}$. 
    Also observe that these agents do not envy the agents whom they did not envy previously since they get better off.
    These two observations give that the new (partial) orientation still satisfies Properties (5)-(7). 
    Moreover, since no agent gets worse off and no edge that was allocated previously becomes unallocated, Properties (1) and (2) are still satisfied. 
    Since $a_{i_1}$ receives only the edge between herself and $a_{i_0}$ who is non-envied, by Property (2), she is safe for all envied agents. 
    Therefore, Property (3) is also satisfied. 
    The (partial) orientation is $\EFX^0_-$ since each agent only receives the edges that are not chores for her by Property (4), and each envied agent only receives one edge by Property (5). 
    Therefore, we complete the proof.  
\end{proof}

When $e_{j, j_t}$ is allocated to $a_{j_t}$, we can also construct a (partial) orientation that is $\EFX^0_-$ and satisfies Properties (1)-(7), following the same process as above. 
Therefore, it remains to consider the case where $e_{i, i_s}$ is not allocated to $a_{i_s}$ and $e_{j, j_t}$ is not allocated to $a_{j_t}$ or such edges do not exist. 
In this case, $a_i$ values $a_{i_s}$'s bundle at zero, and thus $a_{i_s}$ is safe for her. 
This implies that $a_{i_s}$ is not safe for $a_j$; that is, $a_{i_s}$ receives $e_{j, i_s}$ in $\mathbf{X}^1$ and $v_j(\{e_{j, i_s}\} \cup (E_j^{\ge 0} \cap R(\mathbf{X}^1))) > v_j(X_j^1)$. 
Similarly, $a_{j_t}$ is not safe for $a_i$, $a_{j_t}$ receives $e_{i, j_t}$ in $\mathbf{X}^1$ and $v_i(\{e_{i, j_t}\} \cup (E_i^{\ge 0} \cap R(\mathbf{X}^1))) > v_i(X_i^1)$. 
Figure \ref{fig:allocation:EFX0-:step2-2:cases23:initial} illustrates the initial (partial) orientation. 
We further consider the following two cases: 

\begin{figure}[tb]
\centering
\includegraphics[width=0.4\columnwidth]{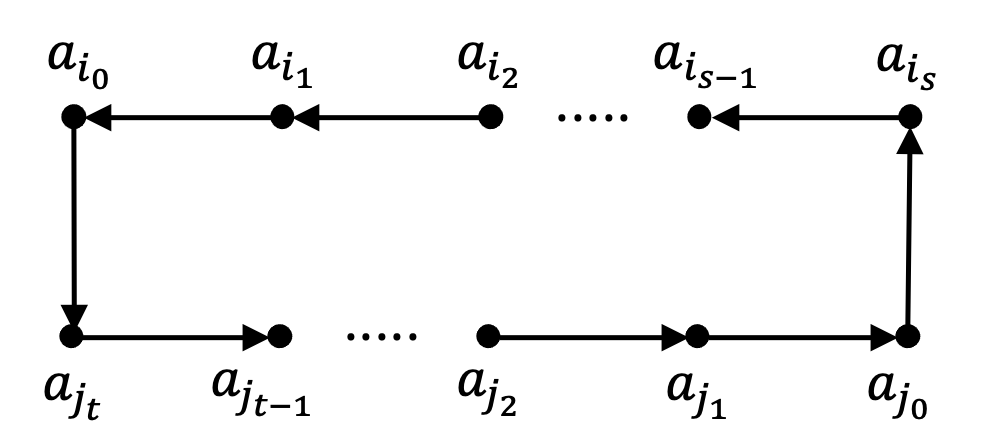}
\caption{Initial (partial) orientation for Cases 2 and 3 of Step 2-2.}
\label{fig:allocation:EFX0-:step2-2:cases23:initial}
\end{figure}

\medskip
\noindent\textbf{\underline{Case 2}:} $a_i$ prefers $e_{i, j}$ to $e_{i, j_t}$ or is indifferent between them

For this case, we let $a_i$ receive the incident edge (if exists) that she values the most among those that are unallocated and are not chores for her (i.e., $E_i^{\ge 0} \cap R(\mathbf{X}^1)$). 
We then allocate $e_{i_{l-1}, i_l}$ to $a_{i_l}$ for every $l \in [s]$ and release $e_{j, i_s}$. 
Figure \ref{fig:allocation:EFX0-:step2-2:case2} illustrates the new (partial) orientation. 
Note that the new (partial) orientation may not satisfy Property (2), since some agents who were non-envied are now envied by $a_i$ and the incident edge $e_{i_s, j}$ of $a_j$ that was allocated now becomes unallocated. 
The good news is that we can prove that the new (partial) orientation is $\EFX^0_-$ and satisfies Properties (1) and (4)-(7). 
Therefore, we can run Algorithm \ref{alg:allocation:EFX0-:property2} to obtain a (partial) orientation that is $\EFX^0_-$ and satisfies Properties (1)-(7). 
We have the following claim. 

\begin{figure}[tb]
\centering
\includegraphics[width=0.47\columnwidth]{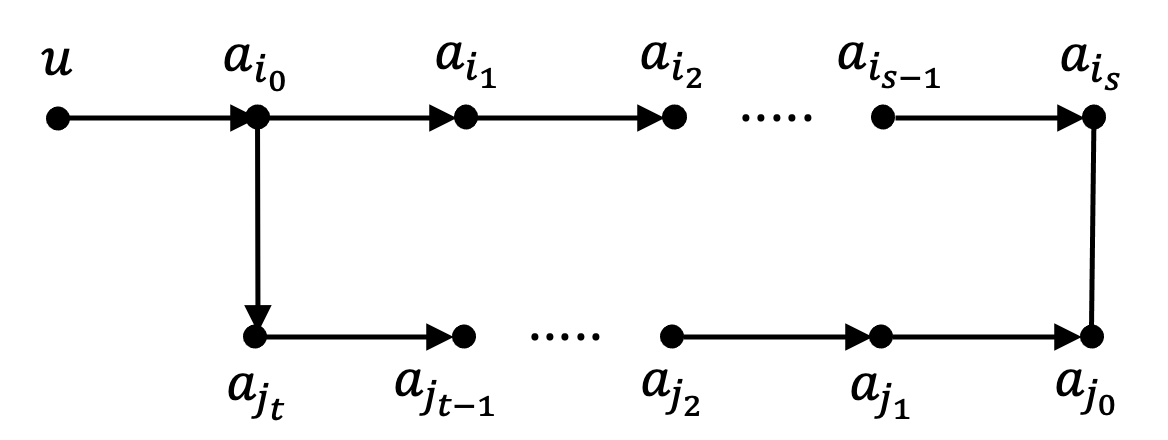}
\caption{The new (partial) orientation for Case 2 of Step 2-2. }
\label{fig:allocation:EFX0-:step2-2:case2}
\end{figure}

\begin{claim}\label{clm:allocation:EFX0-:step2-2:case2}
    The (partial) orientation computed in Case 2 of Step 2-2 is $\EFX^0_-$ and satisfies Properties (1)-(7) in Definition \ref{def:allocation:EFX0-:properties}. 
    Moreover, it is computed in polynomial time. 
\end{claim}
\begin{proof}
    Recall that by Claim \ref{clm:allocation:EFX0-:initial}, the initial (partial) orientation $\mathbf{X}^1$ is $\EFX^0_-$ and satisfies Properties (1) and (4)-(7). 
    We first show that the new (partial) orientation as illustrated in Figure \ref{fig:allocation:EFX0-:step2-2:case2} is also $\EFX^0_-$ and satisfies Properties (1) and (4)-(7). 
    Since no agent receives an edge that is a chore for her in the allocation process, Property (4) holds.
    Moreover, each agent still receives at most one incident edge, thus the new (partial) orientation is $\EFX^0_-$. 
    To see Property (1), it suffices to consider agents $a_i$, $a_j$ and $a_{i_l}$ for $l \in [s]$, since only $a_i$ and $a_{i_l}$ for $l \in [s]$ receive a new bundle in the allocation process, and only $a_j$ has an incident edge that was allocated previously but now becomes unallocated. 
    For $a_i$, if she does not receive an edge, then all her unallocated incident edges are chores for her, thus Property (1) holds for her. 
    If she receives an edge, then the edge is the one she values the most among those in $E_i^{\ge 0} \cap R(\mathbf{X}^1)$, thus Property (1) also holds for her. 
    For each $l \in [s-1]$, $a_{i_l}$ does not get worse off and her unallocated incident edges do not change, thus Property (1) also holds for $a_{i_l}$. 
    For the same reason, both $a_{i_s}$ and $a_j$ prefer their current bundles to each of their unallocated incident edges except $e_{i_s, j}$, which is the only edge that was allocated previously but is now unallocated. 
    For $a_{i_s}$, she prefers her current bundle to $e_{i_s, j}$ since she envied $a_{i_{s-1}}$  previously and now owns the edge $e_{i_{s-1}, i_s}$. 
    For $a_j$, she prefers her current bundle to $e_{i_s, j}$ since she did not envy $a_{i_s}$ previously. 
    Thus, Property (1) also holds for both of them. 
    
    Now consider Properties (5)-(7).
    Since each agent still receives at most one edge, Properties (5) and (6) hold. 
    For Property (7), first observe that except $a_{i_1}$, the agents who receive new bundles are now non-envied (i.e., $a_i$ and $a_{i_l}$ for each $l \in [s] \setminus \{1\}$). 
    To see this, since $\mathbf{X}^1$ satisfies Property (1), the other endpoint agent of the edge that $a_i$ receives (if it exists) does not envy $a_i$. 
    For each $l \in [s] \setminus \{1\}$, $a_{i_l}$ is non-envied by since $a_{i_{l-1}}$ prefers her current bundle to $e_{i_{l-1}, i_l}$. 
    Since $a_i$ is the only agent who may become worse off, the only agents who were non-envied previously and now become envied and who were envied by some agent and are now envied by another agent are all $a_i$'s neighbors. 
    Each of them is envied by $a_i$ who is now non-envied, thus there is still no envy cycle among the agents and Property (7) still holds.
    
    Clearly, the allocation process for Case 2 of Step 2-2 runs in polynomial time. 
    If the new (partial) orientation satisfies Property (2), it also satisfies Property (3). 
    To see this, consider $a_i$ and $a_{j_t}$. 
    Since the edge that $a_i$ owns currently is the one she values the most among those in $E_i^{\ge 0} \cap R(\mathbf{X}^1)$ (including $e_{i, j}$) and she prefers $e_{i, j}$ to $e_{i, j_t}$, she does not envy $a_{j_t}$. 
    Moreover, since $a_{j_t}$ only receives the edge between herself and $a_i$ who is now non-envied, $a_{j_t}$ is safe for all envied agents and Property (3) holds. 
    If the new (partial) orientation does not satisfy Property (2), we run Algorithm \ref{alg:allocation:EFX0-:property2} once. 
    By Claim \ref{clm:allocation:EFX0-:satisfy2}, the final (partial) orientation is $\EFX^0_-$ and satisfies Properties (1)-(7), and is computed in polynomial time. 
\end{proof}

\medskip
\noindent\textbf{\underline{Case 3}:} $a_i$ prefers $e_{i, j_t}$ to $e_{i, j}$

For this case, we allocate all edges in $E_j^{\ge 0} \cap R(\mathbf{X}^1)$ as well as $e_{j, i_s}$ to $a_j$, $e_{j_{l-1}, j_l}$ to $a_{j_l}$ for every $l \in [t]$, and $e_{i_{l-1}, i_l}$ to $a_{i_l}$ for every $l \in [s]$. 
We then release $e_{i, j_t}$ and let $a_i$ receive the incident edge (if it exists) that she values the most among those that are unallocated and are not chores for her. 
Figure \ref{fig:allocation:EFX0-:step2-2:case3} illustrates the new (partial) orientation. 
Similar to Case 2, the new (partial) orientation may not satisfy Property (2), since some agents who were non-envied are now envied by $a_i$. 
But we can prove that the new (partial) orientation is $\EFX^0_-$ and satisfies Properties (1) and (4)-(7). 
Therefore, after running Algorithm \ref{alg:allocation:EFX0-:property2}, we can obtain a (partial) orientation that is $\EFX^0_-$ and satisfies Properties (1)-(7). 
We have the following claim. 

\begin{figure}[tb]
\centering
\includegraphics[width=0.5\columnwidth]{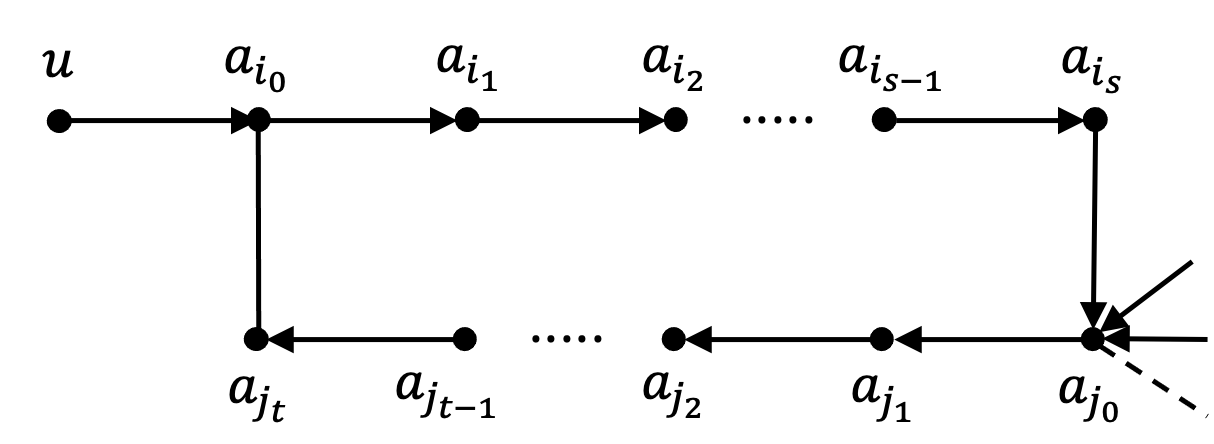}
\caption{The new (partial) orientation for Case 3 of Step 2-2.}
\label{fig:allocation:EFX0-:step2-2:case3}
\end{figure}

\begin{claim}\label{clm:allocation:EFX0-:step2-2:case3}
    The (partial) orientation computed in Case 3 of Step 2-2 is $\EFX^0_-$ and satisfies Properties (1)-(7) in Definition \ref{def:allocation:EFX0-:properties}. 
    Moreover, it is computed in polynomial time. 
\end{claim}
\begin{proof}
    The proof is similar to that of Claim \ref{clm:allocation:EFX0-:step2-2:case2}. 
    Since no agent receives an edge that is a chore for her in the allocation process, Property (4) holds.
    To see Property (1), first consider $a_i$. 
    If $a_i$ does not receive an edge, then all her unallocated incident edges are chores for her, thus Property (1) holds for her. 
    If she receives an edge, the edge is the one that she values the most among her unallocated incident edges that are not chores for her, thus Property (1) also holds for her. 
    For each $l \in [s]$, $a_{i_l}$ gets better off and her unallocated incident edges do not change, thus Property (1) holds for her. 
    For the same reason, Property (1) also holds for $a_{j_l}$ for each $l \in [t-1]$. 
    For $a_j$, all her incident edges that are not chores for her are allocated, thus Property (1) holds for her. 
    For $a_{j_t}$, although $e_{i, j_t}$ was allocated to her previously but is now unallocated, she prefers her current bundle to $e_{i, j_t}$ since she envied $a_{j_{t-1}}$ and now owns the edge $e_{j_t, j_{t-1}}$. 
    Thus Property (1) also holds for $a_{j_t}$. 
    Therefore, the new (partial) orientation retains Property (1). 
    
    To see Properties (5)-(7), first observe that except $a_{i_1}$, the agents who receive new bundles are now non-envied. 
    To see this, since $\mathbf{X}^1$ satisfies Property (1), the other endpoint agent of the edge that $a_i$ receives (if it exists) does not envy $a_i$. 
    For each $l \in [s] \setminus \{1\}$, $a_{i_l}$ is non-envied since $a_{i_{l-1}}$ prefers her current bundle to $e_{i_{l-1}, i_l}$. 
    For the same reason, for each $l \in [t]$, $a_{j_l}$ is non-envied. 
    For $a_j$, $a_{i_s}$ does not envy her since $a_{i_s}$ prefers $e_{i_{s-1}, i_s}$ to $e_{j, i_s}$. 
    The other endpoint agents of the edges in $E_j^{\ge 0} \cap R(\mathbf{X}^1)$ do not envy her, either, since $\mathbf{X}^1$ satisfies Property (1). 
    Since only $a_j$ receives more than one edge and she is now non-envied, Properties (5) and (6) still hold.
    For Property (7), since $a_i$ is the only agent who may become worse off, the only agents who were non-envied previously and now become envied and who were envied by some agent and are now envied by another agent are all $a_i$'s neighbors. 
    Each of them is envied by $a_i$ who is now non-envied, thus there is still no envy cycle among the agents and Property (7) still holds. 
    The (partial) orientation is $\EFX^0_-$ since each agent only receives the edges that are not chores for her by Property (4), and each envied agent only receives one edge by Property (5). 
    
    If the new (partial) orientation satisfies Property (2), it also satisfies Property (3). 
    To see this, consider $a_j$ and $a_{j_1}$. 
    Since $a_{j_1}$ only receives the edge between herself and $a_j$ who is non-envied, $a_{j_1}$ is safe for all envied agents and Property (3) holds. 
    If the new (partial) orientation does not satisfy Property (2), we run Algorithm \ref{alg:allocation:EFX0-:property2} once. 
    By Claim \ref{clm:allocation:EFX0-:satisfy2}, the final (partial) orientation is $\EFX^0_-$ and satisfies Properties (1)-(7), and is computed in polynomial time. 
\end{proof}

Now we are ready to prove Lemma \ref{lem:allocation:EFX0-:first}. 
\begin{proof}[Proof of Lemma \ref{lem:allocation:EFX0-:first}]
    By Claims \ref{clm:allocation:EFX0-:initial}, \ref{clm:allocation:EFX0-:satisfy2}, \ref{clm:allocation:EFX0-:step2-2:case1},  \ref{clm:allocation:EFX0-:step2-2:case2}, \ref{clm:allocation:EFX0-:step2-2:case3}, we know that a (partial) orientation that is $\EFX^0_-$ and satisfies Properties (1)-(7) can be computed in polynomial time. 
    If it also satisfies Property (8), then we have done. 
    
    If the (partial) orientation does not satisfy Property (8), then there exists a sequence of agents $a_{i_0} \leftarrow a_{i_1} \leftarrow \cdots \leftarrow a_{i_s}$ such that  $a_{i_l}$ envies $a_{i_{l-1}}$ for every $l \in [s]$ and $a_{i_s}$ is non-envied, we have that $a_{i_l}$ is not safe for $a_{i_0}$ for some $l \in [s]$. 
    By Property (5), for each $l \in [s-1]$, agent $a_{i_l}$ only receives $e_{i_l, i_{l+1}}$ and does not receive any incident edge of $a_{i_0}$. 
    Thus by Property (2), $a_{i_l}$ is safe for $a_{i_0}$. 
    The only case that Property (8) is not satisfied is that $a_{i_s}$ receives $e_{i_s, i_0}$ and is not safe for $a_{i_0}$. 
    For this case, we run the same allocation process for Case 1 of Step 2-2, as illustrated in Figure \ref{fig:allocation:EFX0-:step2-2:case1}. 
    Specifically, we allocate $e_{i_s, i_0}$ to $a_{i_0}$ as well as all her unallocated incident edges that are not chores for her, and allocate $e_{i_{l-1}, i_l}$ to $a_{i_l}$ for every $l \in [s]$. 
    The above allocation process is repeated until Property (8) is satisfied. 
    By the same reasonings in the proof of Claim \ref{clm:allocation:EFX0-:step2-2:case1}, all the involved agents become non-envied, and $\EFX^0_-$ and Properties (1)-(7) in Definition \ref{def:allocation:EFX0-:properties} are retained. 
    Therefore, the allocation process is repeated for polynomial times and gives a (partial) $\EFX^0_-$ orientation that satisfies Properties (1)-(8). 
\end{proof}

\end{document}